\documentclass{article}

\usepackage{lipsum} 

\usepackage[sc]{mathpazo} 
\usepackage[utf8]{inputenc} 
\usepackage[T1]{fontenc} 
\usepackage{microtype} 

\usepackage[hmarginratio=1:1,top=32mm,columnsep=20pt]{geometry} 
\usepackage[hang, small,labelfont=bf,up,textfont=it,up]{caption} 
\usepackage{booktabs} 
\usepackage{float} 
\usepackage{hyperref} 

\usepackage{lettrine} 
\usepackage{paralist} 

\usepackage{graphicx} 
\usepackage{epic,eepic}
\usepackage[usenames,dvipsnames]{pstricks}
\usepackage{epsfig}
\usepackage{pst-grad} 
\usepackage{pst-plot} 
\usepackage[space]{grffile} 
\usepackage{etoolbox} 

\usepackage[normalem]{ulem} 
\usepackage{cancel}

\usepackage{abstract} 

\usepackage{titlesec} 
\renewcommand\thesection{\Roman{section}} 
\renewcommand\thesubsection{\Roman{subsection}} 
\titleformat{\section}[block]{\large\scshape\centering}{\thesection.}{1em}{} 
\titleformat{\subsection}[block]{\large}{\thesubsection.}{1em}{} 

\usepackage{tikz}
\usetikzlibrary{decorations.markings}
\usetikzlibrary{positioning,chains,fit,shapes,calc}
\definecolor{myblue}{RGB}{80,80,160}
\definecolor{mygreen}{RGB}{80,160,80}
\definecolor{myred}{RGB}{255,0,0}
\definecolor{mybrown}{RGB}{139,69,19}
\usepackage{subfigure}

\usepackage{fancyhdr} 
\usepackage{color} 

\usepackage[reqno]{amsmath}
\usepackage{psfrag, epsfig}
\usepackage{amssymb, amsthm}
\usepackage{graphicx, graphics}
\usepackage{mathrsfs, float}
\usepackage{caption} 
\usepackage{color} 
\usepackage{comment}
\usepackage{multicol}
\usepackage{soul}
\usepackage[normalem]{ulem}

\newtheorem{teorema}{Theorem}

\newtheorem{prop}[teorema]{Proposition}
\newtheorem{lema}[teorema]{Lemma}
\newtheorem{corolario}[teorema]{Corollary}

\newcommand{\cC}{{\mathscr C\!}}

\newcommand{\calC}{\mathcal{C}}
\newcommand{\calV}{\mathcal{V}}

\newcommand{\etal}{\textit{et~al.}}

\newcommand{\tw}{\mathrm{tw}}
\newcommand{\lpt}{\mathrm{lpt}}

\newcommand{\Comp}{\mathrm{Comp}}
\newcommand{\Branch}{\mathrm{Branch}}
\newcommand{\definition}[1]{\textit{#1}}

\usepackage{authblk}

\title{\vspace{-13mm}\fontsize{20pt}{10pt}\selectfont\textbf{
Transversals of longest paths}\thanks{A preliminary version of this paper was presented at the 
IX Latin and American Algorithms, Graphs and Optimization Symposium (LAGOS'17).}} 

\author[1]{M\'{a}rcia R. Cerioli}
\affil[1]{COPPE Sistemas e Instituto de Matem\'{a}tica\\ Universidade Federal do Rio de Janeiro, Brazil}

\author[2]{Cristina G. Fernandes\thanks{Research partially supported by CNPq (Proc.~456792/2014-7 and 308116/2016-0),
    FAPESP (2013/03447-6 and~2015/08538-5) and Project MaCLinC of NUMEC/USP.}} 
\author[2]{Renzo G\'omez\thanks{Research supported by CAPES (Proc.~235671298-48).}}
\author[2]{\\ Juan Guti\'errez\thanks{Research supported by FAPESP (Proc.~2015/08538-5).}}
\affil[2]{Departamento de Ci\^encia da Computa\c c\~ao\\ Universidade de  S\~{a}o Paulo, Brazil}

\author[3]{Paloma T. Lima}
\affil[3]{Department of Informatics\\ University of Bergen, Norway} 
   
\date{}

\begin{document}

\maketitle 

\thispagestyle{fancy} 


\begin{abstract}
  Let $\lpt(G)$ be the minimum cardinality of a set of vertices that intersects all longest paths in a graph~$G$. 
  Let 
  $\omega(G)$ be the size of a maximum clique in $G$,
  and $\tw(G)$ be the treewidth of $G$. We prove that $ \lpt(G)\leq \max \{1, \omega(G)-2\}$
  when $G$ is a connected chordal graph; that $\lpt(G)=1$ when $G$ is a connected
  bipartite permutation graph or a connected full substar graph; and that $\lpt(G)\leq \tw(G)$ for any connected graph $G$.
\end{abstract}


\section{Introduction}

It is a well-known fact that, in a connected graph, any two longest paths have a common vertex.
In 1966, Gallai raised the following question:  \textit{Does every connected graph contain a vertex that belongs to all of its longest paths?}
The answer to Gallai's question is already known to be negative. 
Figure~\ref{waltherCounterexample} shows the smallest known negative example, on 12 vertices, 
which was independently found by Walther and Voss~\cite{Walther74} and Zamfirescu~\cite{Zamfirescu76}.
However, when we restrict ourselves to some specific classes of graphs, the answer to Gallai's question turns out to be positive. 
For example, it is well known that any set of subtrees of a tree satisfies the Helly property. If we consider the set of subtrees consisting of the longest paths of the tree, since they are pairwise intersecting, we conclude that there is a vertex that belongs to all of them. 

There are other graph classes which are known to have a positive answer to Gallai's question.
Klav\v{z}ar and Petkov\v{s}ek~\cite{Klavzar90} proved that this is the case for split graphs, 
cacti, and graphs whose blocks are Hamilton-connected, almost Hamilton-connected or cycles.
Balister \etal~\cite{Balister04} and Joos~\cite{Joos15} proved the same for the class of circular arc graphs. 
De Rezende \etal~\cite{deRezende13} proved that the answer to Gallai's question is positive for 2-trees
and Chen \etal~\cite{Chen17} extended this result for series-parallel graphs, also known as partial 2-trees. 
Chen~\cite{Chen15} proved the same for graphs with matching number smaller than three, 
while Cerioli and Lima~\cite{Cerioli16,Thome16} proved it for $P_4$-sparse graphs, $(P_5,K_{1,3})$-free graphs, graphs that are the join of two other graphs and starlike graphs, 
a superclass of split graphs. Finally, Jobson \etal~\cite{Jobson16} proved it for dually chordal graphs and Golan and Shan~\cite{2K2FREE} for $2K_2$-free graphs.

\begin{figure}[h]
  \centerline{\includegraphics[scale = 1]{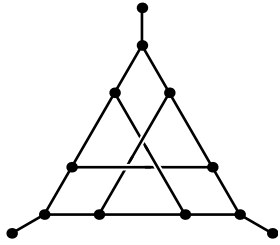}}
  \caption{The classical 12-vertex example that has a negative answer to Gallai's question.}
  \label{waltherCounterexample}
\end{figure}

A more general approach to Gallai's question is to ask for the size of the smallest \emph{transversal of longest paths} of a graph, that is, the smallest set of vertices that intersects every longest path. Given a graph~$G$, we denote the size of such a set by $\lpt(G)$.
In this direction, Rautenbach and Sereni~\cite{RautenbachS14} proved that
$\lpt(G)\leq \lceil \frac{n}{4}-\frac{n^{2/3}}{90} \rceil$ for every connected graph $G$ on $n$ vertices, 
that $\lpt(G)\leq 9 \sqrt{n}\log{n}$ for every connected planar graph $G$ on $n$ vertices,
and that $\lpt(G)\leq k+1$ for every connected graph $G$ of treewidth at most~$k$.
 

In this work, we provide exact results and upper bounds on the value of~$\lpt(G)$ when $G$ belongs to some specific classes of graphs. More specifically, we prove that:  
\begin{itemize}
\item $\lpt(G)\leq \max \{1, \omega(G)-2\}$ for every connected chordal graph~$G$,
where~$\omega(G)$ is the size of a maximum clique of~$G$.
\item $\lpt(G)=1$ for every connected bipartite permutation graph~$G$.
\item $\lpt(G)\leq k$ for every connected graph~$G$ of treewidth at most $k$. 
\item $\lpt(G)=1$ for every connected full substar graph~$G$.
\end{itemize}

This paper is organized as follows. In the next section, we state the definitions and basic results that are going to be used throughout the text.
In Sections~\ref{section:chordals}, \ref{section:bpg}, \ref{section:tw}, and~\ref{section:substars}, we consider, respectively, the class of chordal graphs, the class bipartite permutation graphs, the class of graphs of treewidth at most $k$ and the class of full substar graphs. Finally, in Section~\ref{section:conclusion}, we state the open problems to be considered in future work. 

\section{Definitions and notation}\label{section: definitions}

All graphs considered are simple.
Let $u$ be a vertex in a graph $G$, we denote by $N_G(u)$ the set of neighbors of $u$ in $G$,
and by $d_G(u)$ the cardinality of $N_G(u)$.
If the context is clear, we write simply $d(u)$ and $N(u)$ respectively.
Let~$P$ be a path in a graph~$G$. 
We denote by~$|P|$ the length of~$P$, that is, the number of edges in~$P$.
Given a path~$Q$ such that the only vertex it shares with~$P$ is an extreme 
of both of them, we denote by $P\cdot Q$ the concatenation of~$P$ and~$Q$.
For a vertex~$v$ in~$P$, let~$P'$ and~$P''$ be the paths 
such that $P = P' \cdot P''$ with~$P' \cap P'' = \{v\}$.  
We refer to these two paths as the \emph{$v$-tails} of $P$.
Given a path~$P$ that contains vertices~$a$ and~$b$, 
we denote by $P_a$ the~$a$-tail of~$P$ that does not contain~$b$ 
and by~$P_b$ the $b$-tail of~$P$ that does not contain~$a$. 
Also, if the context is clear, we denote by~$\widetilde{P}$ 
the subpath of~$P$ that has~$a$ and~$b$ as its extremes.
Thus $P = P_a\cdot  \widetilde{P} \cdot P_b$.

Let~$S$ be a set of vertices of~$G$. 
Let~$P$ be a path in~$G$ that does not contain all vertices of $S$ and contains a vertex not in~$S$. 
We say that~$S$ \emph{fences}~$P$ if all the vertices of~$P-S$ are in a single component of $G-S$, 
otherwise we say that~$P$ \emph{crosses}~$S$.
Given a path $P$ that crosses $S$ and has both extremes not in $S$,
we say that~$P$ is \emph{extreme-separated} by~$S$ when the
extremes of~$P$ are in different components of~$G-S$,
and that~$P$ is \emph{extreme-joined} by~$S$ if its extremes are in
the same component of~$G-S$.


For an integer~$t$, we say  that~$P$ \emph{$t$-touches}~$S$ 
if~$P$ intersects~$S$ at exactly~$t$ vertices. A path~$P$ is an
\emph{$S$-corner} path if~$P$ 1-touches~$S$.
Let~$P$ be an	$S$-corner path. If $P$ is fenced by~$S$, we say that~$P$ is an \emph{$S$-corner-fenced} path. 
If~$P$ crosses~$S$, we say that~$P$ is an \emph{$S$-corner-crossing} path.
If two paths~$P$ and $Q$ touch $S$ at the same
set of vertices, we say they are \emph{$S$-equivalent}, otherwise
they are \emph{$S$-nonequivalent}.

If~$P$ is fenced by~$S$, we denote by $\Comp_{S}(P)$ 
the set of vertices of the component of~$G-S$ where~$P-S$ lies.
For a set $X$ of vertices not contained in $S$, we denote by $\Comp_{S}(X)$ 
the set of vertices of the components of~$G-S$ where~$X \setminus S$ lies.
Two fenced paths~$P$ and $Q$ are \emph{$S$-component-disjoint}
if~$\Comp_{S}(P) \neq \Comp_{S}(Q)$. If~$S$ is clear from the context,
we just say they are \emph{component-disjoint}.

From now on, we use~$L=L(G)$ for the length of a longest path in~$G$. 
Also, remember that~$\omega(G)$ is the size of a maximum clique of~$G$.

A graph $H$ is called a \emph{minor} of the graph $G$ if $H$ 
can be formed from $G$ by deleting edges and vertices and by contracting edges.

A \definition{tree decomposition} \cite[p.~337]{Diestel10} of a graph $G$
is a pair $(T, \calV)$, conformed by a tree $T$ and a collection
$\calV = \{ V_t: t \in V(T)\}$ of \definition{bags} $V_t \subseteq V(G)$,
that satisfies the following three conditions:
\begin{enumerate}
\item[(T1)] $\bigcup_{t \in V(T)} V_t = V(G);$
\item[(T2)] for every $uv \in E(G)$, there exists
  a bag $V_t$ such that $u,v \in V_t;$
\item[(T3)] if a vertex $v$ is in two different bags $V_{t_1}, V_{t_2}$,
  then $v$ is also in any bag $V_t$ such that $t$ is on the (unique)
  path from $t_1$ to $t_2$ in $T$.
\end{enumerate}
The \definition{width} of $(T, \calV)$ is the number 
$$\max\{|V_t|-1: t\in V(T)\},$$
and the \definition{treewidth} $tw(G)$ of~$G$ is 
the minimum width of any tree decomposition of~$G$.

A graph is called \textit{chordal} if every induced cycle has length three.
Next we present some basic properties on tree decompositions for general and chordal graphs.
We fix a graph $G$ and a tree decomposition $(T, \calV)$ of $G$.
Proposition~\ref{tree-dec-bodl} is due to Bodlaender~\cite{Bodlaender98}.
Gross~\cite{Gross14} presented a proof for it and refers to tree decompositions 
such as in Proposition~\ref{tree-dec-bodl} as \emph{full tree decompositions}. 
The tree decomposition mentioned in Proposition~\ref{clique-tree}
is also called \emph{clique tree} and it was introduced by Gavril~\cite{Gavril74}.
Proposition~\ref{tw=omega-1} is a direct consequence of
Corollary 12.3.12 of the book of Diestel~\cite{Diestel10}.

\begin{prop}\label{tree-dec-bodl}
  If $k$ is the treewidth of a graph~$G$, 
  then~$G$ has a tree decomposition $(T, \calV)$ of width~$k$ such that
  $|V_t| = k+1$ for every $t \in V(T)$, and $|V_t \cap V_{t'}| = k$ for every $tt' \in E(T)$.
\end{prop}
\begin{prop}\label{clique-tree}
  Every chordal graph $G$ has a tree decomposition~$(T, \mathcal{V})$
  such that the bags of $\mathcal{V}$ are the maximal cliques of $G$.
\end{prop}
\begin{prop}\label{tw=omega-1}
 For every chordal graph $G$, $\tw(G) = \omega(G)-1$.
\end{prop}

Given two different nodes $t$, $t'$ of $T$, we
denote by $\Branch_t(t')$ the component of $T-t$ where~$t'$ lies.
We say that such component is a \emph{branch} of~$T$ at~$t$
and that the components of $T-t$ are the \emph{branches} of~$T$ at~$t$~\cite{Heinz13}.
Similarly, for a vertex $v \notin V_t$, 
we denote by $\Branch_t(v)$ the branch $\Branch_t(t')$ of~$T$ 
at~$t$ such that $v \in V_{t'}$.
We also say that $v$ is \emph{in} $\Branch_t(t')$.
Moreover, we can extend the notation and say that, 
if $P$ is a path fenced by~$V_t$ for some~$t \in T$,
then $\Branch_t(P) = \Branch_t(v)$, where $v$ is a vertex of $P-V_t$.
We also say that $P$ is \emph{in} $\Branch_t(v)$.
Next we show some basic propositions of branches.
Propositions~\ref{one-branch} to~\ref{Branch_t(P)} are used to justify that the previous two
definitions are coherent.
The first two of them appear in the work of Heinz~\cite{Heinz13}.

\begin{prop} \label{one-branch}
  Let $t$ be a node of $T$ and $v$ be a vertex of $G$ such that $v \notin V_t$.
  Let $t'$ and $t''$ be nodes of~$T$.  If $v \in V_{t'} \cap V_{t''}$, 
  then $t'$ and $t''$ are in the same branch of $T$ at $t$.
\end{prop}

\begin{prop} \label{Tu=Tv}
  Let $u$ and $v$ be two vertices of $G$, and let $t$ be a node of~$T$.
  If $u$, $v \notin V_t$, and $u$ and $v$ are not separated by $V_t$ in $G$, 
  then $\Branch_t(u) = \Branch_t(v)$.
\end{prop}

\begin{prop} \label{Branch_t(P)}
  Let $t$ be a node of $T$ and $P$ be a path fenced by $V_t$. 
  For every two vertices $u$ and~$v$ in~$P-V_t$, 
  $\Branch_t(u) = \Branch_t(v)$.
\end{prop}
\begin{proof}
  By definition of fenced paths, $u$ and $v$ lie in the same
  component of $G-V_t$, so they are not separated by $V_t$ in $G$
  and we can apply Proposition~\ref{Tu=Tv}.
\end{proof}

\begin{prop}\label{Branch_t(P)=Branch_t(t')}
  If $P$ is a path fenced by $V_t$ for some $t \in V(T)$,
  then there exists a neighbor $t'$ of $t$ in $T$ such that $\Branch_t(P) = \Branch_t(t')$.
\end{prop}
\begin{proof}
  Let $u$ be a vertex of $P-V_t$ (that exists by the definition of fenced).
  As $u \notin V_t$, there exists a bag $V_{t''} \neq V_{t}$ that contains $u$.
  Let $t'$ be the neighbor of $t$ in $T$ such that $t'$ is in the (unique)
  path from $t$ to $t''$ in $T$.
  Then $\Branch_t(P) = \Branch_t(u) = \Branch_t(t'') = \Branch_t(t')$.
\end{proof}

Proposition~\ref{sep-tt'} appears in the book of Diestel~\cite{Diestel10} as Lemma~12.3.1.
Proposition~\ref{sep-uv} is a corollary of Proposition~\ref{sep-tt'}.

\begin{prop}\label{sep-tt'}
  Let~$pq \in E(T)$.  Let $T_p = \Branch_q(p)$ and $T_q = \Branch_p(q)$ be the components 
  of~$T-pq$, with $p \in V(T_p)$ and~$q \in V(T_q)$. 
  Then $V_p \cap V_q$ separates $\bigcup_{t \in V(T_p)}V_t$ from $\bigcup_{t \in V(T_q)}V_t$ in~$G$.
\end{prop}


\begin{prop}\label{sep-uv}
  Let $pq \in E(T)$.  Let $u$ and $v$ be vertices of $G$ with $v \notin V_p$.
  If~$u \in V_p \setminus V_q$ and ${\Branch_p(v) = \Branch_p(q)}$, 
  then $u$ and $v$ are not adjacent.
\end{prop}

\begin{proof}
  Observe that $u$ is in a bag of $\Branch_q(p)$ because $p$ is in $\Branch_q(p)$ and $u \in V_p$.
  As $v$ is in a bag of $\Branch_p(q)$, by Proposition~\ref{sep-tt'}, 
  $V_p \cap V_q$ separates $u$ from $v$. Hence, $u$ and $v$ are not adjacent.
\end{proof}

\begin{prop} \label{Vt capP'subseteqVt'}
  Let $t \in V(T)$.
  Let $P'$ be a path fenced by $V_t$ that 1-touches $V_t$ such 
  that ${\Branch_t(P')=\Branch_{t}(t')}$, where $tt' \in E(T)$.
  Then $V_t \cap P' \subseteq V_{t'}$.
\end{prop}
\begin{proof}
  Suppose by contradiction
  that there exists a vertex $x	 \in (V_t \cap P') \setminus V_{t'}$.
  As $P'$ is fenced by~$V_t$, there exists a vertex $y \in P'-V_{t}$.
  Moreover, $\Branch_t(y)=\Branch_t(P')=\Branch_t(t')$.
  Let $P'_{xy}$ be the subpath of $P'$ with $x$ and $y$ as extremes.
  Since~$P'$ 1-touches $V_t$, the subpath $P'_{xy}$ also 1-touches $V_t$.
  This implies that $P'_{xy}$ is internally disjoint from $V_t$.
  As $x$, $y \notin V_t \cap V_{t'}$, the subpath $P'_{xy}$ is disjoint from $V_t \cap V_{t'}$.
  But then we contradict Proposition~\ref{sep-tt'}, which says that
  $V_t \cap V_{t'}$ separates $x$, which is in a bag of $\Branch_{t'}(t)$,
  from $y$, which is in a bag of $\Branch_{t}(t')$.
\end{proof}

\begin{prop} \label{V(P')-cap-V(Q')=emptyset}
  Let $t \in V(T)$ and $tt' \in E(T)$.
  Let $P'$ be a path fenced by $V_t$ that 1-touches $V_t$ such that ${\Branch_t(P')=\Branch_{t}(t')}$.
  Let $Q'$ be a path fenced by $V_t$ that 1-touches $V_t$ at a vertex in $V_t \setminus V_{t'}$.
  Then $P' \cap Q' = \emptyset$.
\end{prop}
\begin{proof}
  Suppose by contradiction that $P' \cap Q' \neq \emptyset$.
  As $P \cap Q \cap V_t = \emptyset$, we must have that $\Branch_t(Q')=\Branch_t(P')=\Branch_t(t')$.
  By Proposition~\ref{Vt capP'subseteqVt'}, $V_t \cap Q' \subseteq V_{t'}$, a contradiction.
\end{proof}

\section{Chordal graphs}\label{section:chordals}

We start by proving a lemma that is valid for every graph. 

\begin{lema}\label{extreme-join-clique}
  Let~$G$ be a graph with a clique~$K$.
  Let~$\cC$ be the set of all longest paths in~$G$ that 
  cross~$K$, 2-touch~$K$, and are extreme-joined by~$K$.
  There are at most two~$K$-nonequivalent paths in~$\cC$.
\end{lema}

\begin{proof}
  Suppose by contradiction that there are (at least)
  three $K$-nonequivalent longest paths~$P$,~$Q$, and~$R$ in~$\cC$.
  Say~$P \cap K = \{a,b\}$,~$Q \cap K = \{c,d\}$, and~$R \cap K = \{e,f\}$, 
  where~$\{a, b\}$,~$\{c, d\}$, and~$\{e, f\}$ are pairwise distinct but not necessarily pairwise disjoint.
  We may assume that either $\{a, b\} \cap \{c, d\} = \emptyset$ or $\{a, b\} \cap \{c, d\} = \{b\} = \{d\}$.
  If $\tilde{P}$ is component-disjoint from $Q_c$ (and from $Q_d$), and
  $\tilde{Q}$ is component-disjoint from $P_a$ (and from $P_b$),
  then~$P_a \cdot ac \cdot \tilde{Q} \cdot db \cdot P_b$ and
  $Q_c \cdot ca \cdot \tilde{P} \cdot bd \cdot Q_d$ are paths 
  whose lengths sum more than~$2L$, a contradiction, 
  as at least one of them would have length greater than $L$.
  So, 
  \begin{equation}\label{PconQ}
    \Comp_{K}(\tilde{P})=\Comp_{K}(Q_c) \mbox{ \ or \ } \Comp_{K}(\tilde{Q})=\Comp_{K}(P_a),
  \end{equation}
  Applying the same reasoning to paths $P$ and $R$, and to paths $Q$ and $R$, we conclude that 
  \begin{equation}\label{PconR}
    \Comp_{K}(\tilde{P})=\Comp_{K}(R_e)\mbox{ \ or \ } \Comp_{K}(\tilde{R})=\Comp_{K}(P_a),
 \end{equation}
 and that 
  \begin{equation}\label{QconR}
    \Comp_{K}(\tilde{Q})=\Comp_{K}(R_e) \mbox{ \ or \ } \Comp_{K}(\tilde{R})=\Comp_{K}(Q_c).
 \end{equation}
 Also, as $P$, $Q$, and $R$ cross $K$, from (\ref{PconQ}), (\ref{PconR}), and (\ref{QconR}), we have that
 \begin{equation}\label{diferentes}
 \Comp_{K}(\tilde{P}) \neq \Comp_{K}(\tilde{Q}) \mbox{, \ }
 \Comp_{K}(\tilde{P}) \neq \Comp_{K}(\tilde{R})\mbox{, \ and \ }
 \Comp_{K}(\tilde{Q}) \neq \Comp_{K}(\tilde{R}).
 \end{equation}
  
  Without loss of generality, we may assume that $\Comp_{K}(\tilde{P})=\Comp_{K}(Q_c)$. 
  (Otherwise, interchange $P$ with~$Q$, and $\{a,b\}$ with $\{c,d\}$.)
  See Figure~\ref{fig:threecomponents}(a).
  Now, if $\Comp_{K}(\tilde{P})=\Comp_{K}(R_e)$, then, by (\ref{diferentes}),
  $\Comp_{K}(\tilde{Q}) \neq \Comp_{K}(R_e)$, 
  and thus, by (\ref{QconR}), $\Comp_{K}(\tilde{R})=\Comp_{K}(Q_c)$.
  But then  
  $\Comp_{K}(\tilde{R})=\Comp_{K}(\tilde{P})$, and we contradict (\ref{diferentes}).
  Hence, $\Comp_{K}(\tilde{P})\neq \Comp_{K}(R_e)$, and, by (\ref{PconR}),
  $\Comp_{K}(\tilde{R})=\Comp_{K}(P_a)$.
  Similarly, one can deduce that $\Comp_{K}(\tilde{R}) \neq \Comp_{K}(Q_c)$. 
  Thus, by~(\ref{QconR}), $\Comp_{K}(\tilde{Q})=\Comp_{K}(R_e)$,
  and, again, we can deduce that $\Comp_{K}(\tilde{Q}) \neq \Comp_{K}(P_a)$.
  As~$P$, $Q$, and $R$ are extreme-joined, we conclude that 
  \begin{equation} \label{P_aP_btildeQ}
  \Comp_{K}(P_a)=\Comp_{K}(P_b) \neq \Comp_{K}(\tilde{Q}),
  \end{equation}
  \begin{equation} \label{Q_cQ_dtildeR}
  \Comp_{K}(Q_c)=\Comp_{K}(Q_d) \neq  \Comp_{K}(\tilde{R}),
  \end{equation}
  \begin{equation} \label{R_eR_ftildeP}
  \Comp_{K}(R_e)=\Comp_{K}(R_f) \neq  \Comp_{K}(\tilde{P}).
  \end{equation}
  See Figure~\ref{fig:threecomponents}(b).
 
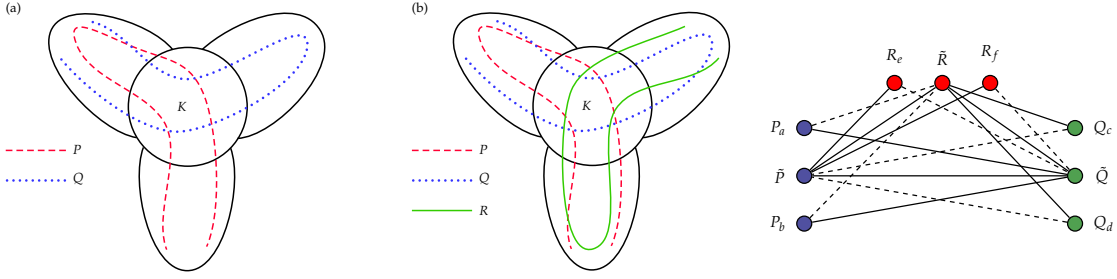
\begin{figure}[thb]
\parbox{10cm}{
\centering
\scalebox{.5}{
%
%
\psscalebox{1.0 1.0} 
{
\begin{pspicture}(0,-3.537896)(19.327814,3.537896)
\definecolor{colour0}{rgb}{1.0,0.0,0.2}
\definecolor{colour1}{rgb}{0.2,0.2,1.0}
\definecolor{colour2}{rgb}{0.2,0.8,0.0}
\pscircle[linecolor=black, linewidth=0.04, dimen=outer](15.635,0.86424345){1.6}
\rput[bl](15.355,0.7642435){$K$}
\psbezier[linecolor=black, linewidth=0.04](16.015,2.3842435)(16.746798,3.0657651)(18.549599,3.8783824)(19.155,2.9042434692382812)(19.7604,1.9301044)(17.954262,0.06267659)(16.955,0.024243468)
\psbezier[linecolor=black, linewidth=0.04](15.235,2.3842435)(14.608967,3.1640396)(12.377352,3.8082602)(11.995,2.8842434692382812)(11.612648,1.9602267)(13.287636,0.20270936)(14.275,0.04424347)
\psbezier[linecolor=black, linewidth=0.04](14.455,-0.21575654)(14.1045065,-1.1523216)(14.615035,-3.4874234)(15.615,-3.4957565307617187)(16.614965,-3.5040896)(17.113516,-1.176717)(16.775,-0.23575653)
\psbezier[linecolor=colour0, linewidth=0.04, linestyle=dashed, dash=0.17638889cm 0.10583334cm](15.075,-2.9357564)(14.698983,-2.0091436)(15.517241,-0.4204843)(15.015,0.44424346923828123)(14.512759,1.3089713)(13.078966,1.5608963)(12.695,2.4842434)(12.311034,3.4075906)(13.250983,2.8465953)(14.175,2.4642434)(15.099017,2.0818918)(15.675613,2.1294818)(16.055,1.2042434)(16.434387,0.27900514)(16.559065,-1.950125)(16.135,-2.8557565)
\psbezier[linecolor=colour1, linewidth=0.07, linestyle=dotted, dotsep=0.10583334cm](12.275,2.1242435)(12.910523,1.3521616)(14.416512,0.1492808)(15.415,0.20424346923828124)(16.413488,0.25920615)(18.743753,1.4867848)(18.815,2.4842434)(18.886248,3.481702)(16.50863,1.4915637)(15.515,1.6042435)(14.521369,1.7169234)(13.982365,2.9857776)(12.995,3.1442435)
\psbezier[linecolor=colour2, linewidth=0.04](18.095,2.9842434)(17.181992,2.5763035)(15.456064,2.2087922)(15.075,1.2842434692382811)(14.693936,0.3596948)(14.700853,-3.2350967)(15.655,-2.9357564)(16.609146,-2.6364164)(15.705626,-0.051332574)(16.255,0.78424346)(16.804375,1.6198195)(18.353445,1.3771666)(18.995,2.1442435)
\psline[linecolor=colour2, linewidth=0.04](10.835,-1.9157566)(12.435,-1.9157566)
\psline[linecolor=colour1, linewidth=0.07, linestyle=dotted, dotsep=0.10583334cm](10.835,-1.1157565)(12.435,-1.1157565)
\psline[linecolor=colour0, linewidth=0.04, linestyle=dashed, dash=0.17638889cm 0.10583334cm](10.835,-0.31575653)(12.435,-0.31575653)
\rput[bl](12.615,-0.41575652){$P$}
\rput[bl](12.615,-1.2357565){$Q$}
\rput[bl](12.615,-2.0357566){$R$}
\pscircle[linecolor=black, linewidth=0.04, dimen=outer](4.875,0.84424347){1.6}
\rput[bl](4.595,0.74424344){$K$}
\psbezier[linecolor=black, linewidth=0.04](5.255,2.3642435)(5.986798,3.0457652)(7.7895994,3.8583825)(8.395,2.8842434692382812)(9.000401,1.9101044)(7.194261,0.04267659)(6.195,0.0042434693)
\psbezier[linecolor=black, linewidth=0.04](4.475,2.3642435)(3.8489664,3.1440396)(1.6173518,3.7882602)(1.235,2.864243469238281)(0.8526482,1.9402267)(2.5276356,0.18270937)(3.515,0.024243468)
\psbezier[linecolor=black, linewidth=0.04](3.695,-0.23575653)(3.3445063,-1.1723217)(3.8550348,-3.5074234)(4.855,-3.5157565307617187)(5.854965,-3.5240896)(6.353516,-1.196717)(6.015,-0.25575653)
\psbezier[linecolor=colour0, linewidth=0.04, linestyle=dashed, dash=0.17638889cm 0.10583334cm](4.315,-2.9557564)(3.9389832,-2.0291436)(4.757241,-0.4404843)(4.255,0.42424346923828127)(3.7527592,1.2889712)(2.3189662,1.5408963)(1.935,2.4642434)(1.5510339,3.3875906)(2.4909832,2.8265953)(3.415,2.4442434)(4.339017,2.0618918)(4.9156137,2.1094818)(5.295,1.1842434)(5.6743865,0.25900516)(5.7990656,-1.9701251)(5.375,-2.8757565)
\psbezier[linecolor=colour1, linewidth=0.07, linestyle=dotted, dotsep=0.10583334cm](1.515,2.1042435)(2.1505232,1.3321617)(3.6565115,0.1292808)(4.655,0.18424346923828125)(5.6534886,0.23920614)(7.9837527,1.4667847)(8.055,2.4642434)(8.126247,3.461702)(5.7486315,1.4715636)(4.755,1.5842434)(3.7613688,1.6969233)(3.2223644,2.9657776)(2.235,3.1242435)
\psline[linecolor=colour1, linewidth=0.07, linestyle=dotted, dotsep=0.10583334cm](0.035,-1.1157565)(1.635,-1.1157565)
\psline[linecolor=colour0, linewidth=0.04, linestyle=dashed, dash=0.17638889cm 0.10583334cm](0.035,-0.31575653)(1.635,-0.31575653)
\rput[bl](1.815,-0.41575652){$P$}
\rput[bl](1.815,-1.2357565){$Q$}
\rput[bl](0.035,3.2842436){(a)}
\rput[bl](10.835,3.2842436){(b)}
\end{pspicture}
}}
}
\parbox{5cm}{
\centering
\scalebox{.6}{
\begin{tikzpicture}[thick,
  every node/.style={draw,circle},
  psnode/.style={fill=myblue},
  qsnode/.style={fill=mygreen},
  rsnode/.style={fill=myred},
  style1/.style={ellipse,draw,inner sep=-1pt,text width=1.5cm},
  style2/.style={ellipse,draw,inner sep=0pt,text width=4cm,text heigth=16cm,ellipse ratio=2},
]

\begin{scope}[start chain=going below,node distance=7mm,yshift=1.5cm]
\node[psnode,on chain] (pa) [label=left: $P_a$] {};
\node[psnode,on chain] (tp) [label=left: $\tilde{P}$] {};
\node[psnode,on chain] (pb) [label=left: $P_b$] {};
\end{scope}

\begin{scope}[xshift=6cm,yshift=1.5cm,start chain=going below,node distance=7mm]
\node[qsnode,on chain] (qc) [label=right: $Q_c$] {};
\node[qsnode,on chain] (tq) [label=right: $\tilde{Q}$] {};
\node[qsnode,on chain] (qd) [label=right: $Q_d$] {};
\end{scope}

\begin{scope}[xshift=2cm,yshift=2.5cm,start chain=going right,node distance=7mm]
\node[rsnode,on chain] (re) [label=above: $R_e$] {};
\node[rsnode,on chain] (tr) [label=above: $\tilde{R}$] {};
\node[rsnode,on chain] (rf) [label=above: $R_f$] {};
\end{scope}

\draw (tp) -- (tq);
\draw (tq) -- (tr);
\draw (tr) -- (tp);
\draw (pa) -- (tq);
\draw (pb) -- (tq);
\draw (qc) -- (tr);
\draw (qd) -- (tr);
\draw (re) -- (tp);
\draw (rf) -- (tp);

\draw [dashed]  (re) -- (tq);
\draw [dashed]  (rf) -- (tq);
\draw [dashed]  (qc) -- (tp);
\draw [dashed]  (qd) -- (tp);
\draw [dashed]  (pa) -- (tr);
\draw [dashed]  (pb) -- (tr);

\end{tikzpicture}
}}
\caption{(a) Paths $P$ and $Q$ from the proof of Lemma~\ref{extreme-join-clique}. (b) Paths $P$, $Q$, and $R$, 
and a tripartite graph representing the interaction between their parts. The vertex set of the graph has three 
vertices for each of the paths, one for each part.  The edge set contains two types of edges: the straight edges 
connect parts that are component-disjoint and the dashed edges connect parts that are not component-disjoint.
When the interaction between two parts is not determined, we omit the edge between them.}
\label{fig:threecomponents}
\end{figure}

  Hence, by (\ref{P_aP_btildeQ}), (\ref{Q_cQ_dtildeR}), and (\ref{R_eR_ftildeP}),
  we have three paths, 
  $P_a \cdot ac \cdot \widetilde{Q} \cdot db \cdot P_b$, \ 
  $Q_c \cdot ce \cdot \widetilde{R} \cdot fd \cdot Q_d$, \
  and \ $R_e \cdot ea \cdot \widetilde{P} \cdot bf \cdot R_f$,
  whose lengths sum more than~$3L$, which leads to a contradiction.
\end{proof}

The previous lemma examines how longest paths that are extreme-joined by a clique behave. 
The following lemma examines the case in which the longest paths are extreme-separated. 
Observe that, in both cases, we are only considering longest paths that cross the clique and touch it at most twice.

\begin{lema}\label{extremeseparated}
  Let~$G$ be a graph with a clique~$K$ and let~$\cC$ be the set of all 
  longest paths in~$G$ that are extreme-separated and touch~$K$ at most twice.
  Every two elements of~$\cC$ have a common vertex of~$K$.
\end{lema}

\begin{proof}
  Let~$P$ and~$Q$ be two arbitrary paths in~$\cC$.
  Suppose by contradiction that $P \cap Q \cap K = \emptyset$.
  As~$P$ is extreme-separated by $K$, path $P$ crosses $K$ 
  and therefore $P$ either 1-touches or 2-touches $K$.  
  To address these two possibilities at once, 
  let $x$ and $y$ be such that $P$ touches $K$ at $x$ and $y$, 
  with $x \neq y$ if~$P$ 2-touches $K$.  
  Also, if $x = y$, then let $P_x$ and $P_y$ be different $x$-tails of $P$ 
  and let $\widetilde{P}$ be the path consisting of only the vertex $x$. 
  Let~$w$ and~$z$, and possibly $Q_w$, $Q_z$, and $\widetilde{Q}$, 
  be defined similarly for~$Q$. 
  
  As both~$P$ and~$Q$ are extreme-separated by~$K$, the tail 
  $P_x$ is component-disjoint from at least one in $\{Q_w,Q_z\}$.
  Analogously,~$P_y$ is component-disjoint from at least one in $\{Q_w,Q_z\}$,
  $Q_w$ is component-disjoint from at least one in $\{P_x, P_y\}$ and
  $Q_z$ is component-disjoint from at least one in $\{P_x, P_y\}$.
  We may assume without loss of generality that~$P_x$ and~$Q_w$ are 
  component-disjoint and that~$P_y$ and~$Q_z$ are component-disjoint.
  (Otherwise interchange~$w$ and~$z$.)  Observe also that~$\widetilde{P}$ 
  is component-disjoint from at least one in~$\{Q_w, Q_z\}$.  Without loss 
  of generality, assume that~$\widetilde{P}$ is component-disjoint from~$Q_w$.  
  (Otherwise interchange $x$ and~$y$, and $w$ and~$z$ simultaneously.)
  
  Note that~$\widetilde{Q}$ is component-disjoint from at least one in~$\{P_x, P_y\}$.
  First suppose that~$\widetilde{Q}$ is component-disjoint from~$P_y$.
  (See a representation of the interactions between the parts of $P$ and~$Q$
  in Figure \ref{fig:extremeseparated}(a).)
  Then, one of the paths~$P_x \cdot \widetilde{P} \cdot yw \cdot Q_w$ or
  $P_y \cdot yw \cdot \widetilde{Q} \cdot Q_z$ is longer than~$L$, a contradiction.
  Now suppose that~$\widetilde{Q}$ is not component-disjoint from~$P_y$,
  that is, ${\Comp_K(\widetilde{Q}) = \Comp_K(P_y)}$,
  and thus~$\widetilde{Q}$ is component-disjoint from~$P_x$.
  If $\widetilde{P}$ and $\widetilde{Q}$ are component-disjoint
  (see Figure \ref{fig:extremeseparated}(b)), then one of
  $P_x \cdot \widetilde{P} \cdot yz \cdot \widetilde{Q} \cdot Q_w$ 
  or~$P_y \cdot yz \cdot Q_z$ is longer than~$L$, a contradiction.
  If $\widetilde{P}$ and $\widetilde{Q}$ are not component-disjoint,
  that is, ${\Comp_K(\widetilde{P}) = \Comp_K(\widetilde{Q})}$, then,
  as ${\Comp_K(\widetilde{Q}) = \Comp_K(P_y)}$ and $P_y$ and $Q_z$
  are component-disjoint, we have that $Q_z$ is component-disjoint
  from~$\widetilde{P}$ (see Figure \ref{fig:extremeseparated}(c)).
  Thus, one of the paths~$P_x \cdot xz \cdot \widetilde{Q} \cdot Q_w$
  or~$P_y  \cdot \widetilde{P} \cdot xz \cdot Q_z$ is longer than~$L$,
  a contradiction. 
\end{proof}

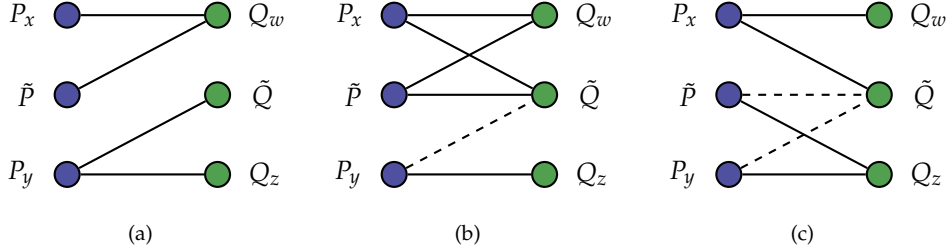
\begin{figure}[ht]
\centering

\subfigure[ ]{
\begin{tikzpicture}[thick,
  every node/.style={draw,circle},
  psnode/.style={fill=myblue},
  qsnode/.style={fill=mygreen},
  rsnode/.style={fill=myred},
  style1/.style={ellipse,draw,inner sep=-1pt,text width=1.5cm},
  style2/.style={ellipse,draw,inner sep=0pt,text width=4cm,text heigth=16cm,ellipse ratio=2},
]

\begin{scope}[start chain=going below,node distance=7mm,yshift=1.5cm]
\node[psnode,on chain] (px) [label=left: $P_x$] {};
\node[psnode,on chain] (tp) [label=left: $\tilde{P}$] {};
\node[psnode,on chain] (py) [label=left: $P_y$] {};
\end{scope}

\begin{scope}[xshift=2cm,yshift=1.5cm,start chain=going below,node distance=7mm]
\node[qsnode,on chain] (qw) [label=right: $Q_w$] {};
\node[qsnode,on chain] (tq) [label=right: $\tilde{Q}$] {};
\node[qsnode,on chain] (qz) [label=right: $Q_z$] {};
\end{scope}

\draw (px) -- (qw);
\draw (tp) -- (qw);
\draw (py) -- (tq);
\draw (py) -- (qz);
\end{tikzpicture}
}
\subfigure[ ]{
\begin{tikzpicture}[thick,
  every node/.style={draw,circle},
  psnode/.style={fill=myblue},
  qsnode/.style={fill=mygreen},
  rsnode/.style={fill=myred},
  style1/.style={ellipse,draw,inner sep=-1pt,text width=1.5cm},
  style2/.style={ellipse,draw,inner sep=0pt,text width=4cm,text heigth=16cm,ellipse ratio=2},
]

\begin{scope}[start chain=going below,node distance=7mm,yshift=1.5cm]
\node[psnode,on chain] (px) [label=left: $P_x$] {};
\node[psnode,on chain] (tp) [label=left: $\tilde{P}$] {};
\node[psnode,on chain] (py) [label=left: $P_y$] {};
\end{scope}

\begin{scope}[xshift=2cm,yshift=1.5cm,start chain=going below,node distance=7mm]
\node[qsnode,on chain] (qw) [label=right: $Q_w$] {};
\node[qsnode,on chain] (tq) [label=right: $\tilde{Q}$] {};
\node[qsnode,on chain] (qz) [label=right: $Q_z$] {};
\end{scope}

\draw (tp) -- (qw);
\draw (tq) -- (px);
\draw (tp) -- (tq);
\draw (py) -- (qz);
\draw (px) -- (qw);
\draw [dashed]  (py) -- (tq);
\end{tikzpicture}
}
\subfigure[ ]{
\begin{tikzpicture}[thick,
  every node/.style={draw,circle},
  psnode/.style={fill=myblue},
  qsnode/.style={fill=mygreen},
  rsnode/.style={fill=myred},
  style1/.style={ellipse,draw,inner sep=-1pt,text width=1.5cm},
  style2/.style={ellipse,draw,inner sep=0pt,text width=4cm,text heigth=16cm,ellipse ratio=2},
]

\begin{scope}[start chain=going below,node distance=7mm,yshift=1.5cm]
\node[psnode,on chain] (px) [label=left: $P_x$] {};
\node[psnode,on chain] (tp) [label=left: $\tilde{P}$] {};
\node[psnode,on chain] (py) [label=left: $P_y$] {};
\end{scope}

\begin{scope}[xshift=2cm,yshift=1.5cm,start chain=going below,node distance=7mm]
\node[qsnode,on chain] (qw) [label=right: $Q_w$] {};
\node[qsnode,on chain] (tq) [label=right: $\tilde{Q}$] {};
\node[qsnode,on chain] (qz) [label=right: $Q_z$] {};
\end{scope}

\draw (tq) -- (px);
\draw (py) -- (qz);
\draw (px) -- (qw);
\draw (tp) -- (qz);
\draw [dashed]  (py) -- (tq);
\draw [dashed]  (tp) -- (tq);
\end{tikzpicture}
}

\caption{Each bipartite graph represents the situation of the paths $P$ and $Q$ in one 
of the cases of the proof of Lemma~\ref{extremeseparated}.  Each side of the bipartition 
has three vertices that represent the parts of each path.  There is a straight edge in the graph 
if the two corresponding vertices are component-disjoint and a dashed edge if they are not.}
\label{fig:extremeseparated}
\end{figure}

The following lemma synthesizes the two previous lemmas.
It says that, for every clique, when the transversal is not in it,
we would have a longest path that is
fenced by the clique. Observe that the lemma is valid only for
chordal graphs. Remember that $\omega(G)$ is the size
of a maximum clique in $G$. A \emph{$k$-clique} is a subset of $k$ vertices
in $G$ that are pairwise adjacent.

\begin{lema}\label{fenceds}
  Let~$G$ be a connected chordal graph with a~$k$-clique~$K$. One of the following is true:
  \begin{itemize}
  \item[(a)] $\lpt(G) \leq \max \{1, \omega(G)-2\}$.
  \item[(b)] There exists a longest path that  does not touch $K$.
  \item[(c)] There exists a vertex $v$ of $K$ such that
    there is a longest path that is fenced by $K$ and 1-touches $K$ at~$v$.
    Moreover, no longest path that 1-touches $K$ at $v$ crosses $K$.
  \item[(d)] There exists an edge $e$ of $K$ such that
    there is a longest path that is fenced by $K$ and 2-touches $K$ at the ends of $e$.
    Moreover, no longest path that 2-touches $K$ at the ends of $e$ crosses $K$.
    \end{itemize}
\end{lema}

\begin{proof}
  We will prove that the negation of $(a)$, $(b)$, and $(c)$ implies $(d)$.
  So, suppose that no clique of size~$\max \{1, \omega(G)-2\}$ is a longest 
  path transversal in~$G$, that every longest path touches $K$ at least once,
  and that, if a vertex $v$ is such that some longest path $1$-touches $K$ at~$v$, 
  then there exists a longest path that  $1$-touches $K$ at $v$ and crosses $K$.  
  If~${k\leq \max \{1, \omega(G)-2\}}$, then either~$(a)$ or~$(b)$ holds.
  So we may assume that~$k \ge \max\{1, \omega(G)-2\} + 1$.  
  If $\omega(G) \le 3$, 
  then $\tw(G)\leq 2$ by Proposition \ref{tw=omega-1}, and $(a)$ holds by Chen \etal~\cite{Chen17}.
  We conclude that $\omega(G) \geq 4$ and $k \in \{\omega(G)-1,\omega(G)\}$.
  
  \medskip

  \noindent \emph{Case 1.} There is a longest path that 1-touches~$K$.

  \smallskip

  If $k = \omega(G)-1$ then, 
  as~$(a)$ and~$(b)$ do not hold, for every vertex in $K$, 
  there exists a longest path that 1-touches $K$ at that vertex. 
  Also, as~$(c)$ is false, we may assume that each such path crosses~$K$, 
  a contradiction to Lemma~\ref{extremeseparated}, because $k\geq 3$.
  So $k = \omega(G)$.
  As~$(c)$ does not hold, and we are assuming that there is a longest path
  that 1-touches $K$, there exists a longest path~$P$
  that 1-touches~$K$ at a vertex~$v$ and crosses~$K$.  
  As~$(a)$ does not apply, for every $(k-2)$-clique in~$K$ containing $v$, 
  there exists a longest path that does not contain any vertex in that clique.
  If any of these longest paths 1-touches $K$ at a vertex $w$, then, as~$(c)$ does not hold, 
  there is a longest path that crosses $K$ at $w$,
  contradicting Lemma~\ref{extremeseparated}.
  Hence, for every edge in $K$ not incident to $v$, there exists a longest 
  path that 2-touches~$K$ at the ends of that edge. 
  Again, by Lemma~\ref{extremeseparated}, as $P$ crosses $K$ at $v$,
  none of these paths is extreme-separated by~$K$. 
  As~$k \geq  4$, there are at least three such paths.  
  By Lemma~\ref{extreme-join-clique}, one of these edges, call it~$e$,
  is such that no longest path crosses~$K$ and 2-touches~$K$ at the ends of~$e$.
  Moreover, we know that there is a longest path that 2-touches~$K$ at the ends of~$e$ and,
  by the previous discussion, that path is fenced by~$K$. So~$(d)$ holds.

  \medskip
  
  \noindent \emph{Case 2.} Every longest path touches~$K$ at least twice.
  
  \smallskip

  If $k = \omega(G)-1$, then any subset of vertices of $K$ of size $\omega(G)-2$ is a longest path transversal, 
  and $(a)$ would hold.   Thus we may also assume that $k = \omega(G)$.
  As $\lpt(G)>\omega(G)-2=k-2$, for every edge of~$K$, there exists a longest path that 2-touches~$K$ at the ends of that edge.
  As $k\geq 4$, there are at least six nonequivalent longest paths that 2-touch~$K$.
  Suppose by contradiction that $(d)$ does not hold. 
  Hence, we may assume that these six paths cross $K$.
  By Lemma~\ref{extreme-join-clique}, four of 
  these longest paths are extreme-separated by~$K$.  
  As at least two of the corresponding four edges of~$K$ are disjoint, 
  by Lemma~\ref{extremeseparated}, we have a contradiction.
\end{proof}


\medskip
  
We can finally prove our main result.

\begin{teorema}\label{lptchordal}
  For every connected chordal graph~$G$,~$\lpt(G)\leq \max \{1, \omega(G)-2\}.$
\end{teorema}

\begin{proof}
  Suppose by contradiction that~$\lpt(G) > \max \{1, \omega(G)-2\}$.
  Then, for every clique~$K$ in~$G$, there exists a longest path 
  fenced by~$K$ as in~$(b)$,~$(c)$, or~$(d)$ of Lemma~\ref{fenceds}.
  We create a directed graph $D$, 
  that admits antiparalell arcs, as follows.  
  Let~$(T, \calV)$ be a tree decomposition of~$G$. 
  The nodes of $D$ are exactly the nodes of $T$.
  Let $t$ be a node in $T$ and let $P$ be a longest path in $G$
  fenced by $V_t$ that satisfies one 
  of the conditions $(b)$, $(c)$, or $(d)$ of Lemma~\ref{fenceds}.
  By Proposition~\ref{Branch_t(P)=Branch_t(t')}, 
  there exists a neighbor $t'$ of $t$ in $T$ such that $\Branch_t(P)=\Branch_t(t')$.
  Hence $tt'$ is an arc in $D$.
  Thus every node of~$D$ is the tail of at least one arc in $D$.
  
  Let $tt'$ be the last arc of a maximal directed path in~$D$.
  As~$T$ is a tree,~$t't$ is also an arc in~$D$, which implies 
  that there exist two longest paths~$P$ and~$Q$ in $G$ such that
  $\Branch_t(P)=\Branch_t(t')$ and $\Branch_{t'}(Q)=\Branch_{t'}(t)$,
  where~$P$ is fenced by~$V_t$ and~$Q$ is fenced by~$V_{t'}$, and
  both satisfy one of the conditions~$(b)$, $(c)$, or~$(d)$ of 
  Lemma~\ref{fenceds}.
  
  From now on, we assume that~$(T, \calV)$ is a 
  tree decomposition of~$G$ as in Proposition~\ref{clique-tree}.
  Note that the bags containing vertices of $P$ are only 
  in $\Branch_t(t') \cup \{t\}$, and the bags containing
  vertices of $Q$ are only in $\Branch_{t'}(t) \cup \{t'\}$.
  As $\Branch_t(t')$ and $\Branch_{t'}(t)$ are disjoint, 
  ${P \cap Q \subseteq V_{t} \cup V_{t'}}$.  
  Let $u$ be a vertex such that $u \in V_{t}\setminus V_{t'}$.
  Suppose for a moment that $P$ contains~$u$ and 
  let $v$ be a neighbor of $u$ in $P$.
  By Proposition~\ref{sep-uv}, vertex $v$ cannot be in $\Branch_{t}(t')$, 
  so~$v \in V_{t}$.  This implies that $uv$ is an edge in $V_{t}$
  and, as $V_{t}$ is a clique, $P$ contains all vertices of~$V_{t}$,
  contradicting the fact that $P$ is fenced.
  So $P$ does not contain vertices in $V_{t} \setminus V_{t'}$.
  By a similar reasoning, $Q$ does not contain vertices in $V_{t'} \setminus V_t$.
  Thus $P \cap Q \subseteq V_{t}\cap V_{t'}$.  As $G$ is connected,
  \begin{equation}\label{PQVtVt'}
    P \cap Q \ = \ P \cap Q \cap V_{t}\cap V_{t'} \ \neq \ \emptyset.
  \end{equation}
  This implies that $P \cap V_{t}\neq \emptyset$
  and $Q \cap V_{t'}\neq \emptyset$,
  therefore none of $P$ and $Q$ satisfies condition~$(b)$ of Lemma~\ref{fenceds}.
  
  Suppose for a moment that $|V_{t}\cap V_{t'}| \leq \omega(G)-2$.
  Then, as $\lpt(G) > \max \{1, \omega(G)-2\}$, there exists a longest path $R$
  that does not contain any vertex of $V_{t}\cap V_{t'}$.
  As $G$ is connected, $R$ intersects~$P$.
  As $P$ does not contain vertices in $V_{t} \setminus V_{t'}$ and 
  $R$ does not contain vertices in $V_{t}\cap V_{t'}$, we have that $P \cap R \nsubseteq V_t$. 
  As the bags containing vertices of $P$ are only 
  in $\Branch_t(t') \cup \{t\}$, $R$ has a vertex in a bag of $\Branch_t(t')$.
  A similar reasoning, with $Q$ instead of $P$, shows that $R$ also
  has a vertex in a bag of $\Branch_{t'}(t)$.
  This is a contradiction to Proposition \ref{sep-tt'}, as $R$ contains no vertex 
  in~$V_{t}\cap V_{t'}$. Hence $|V_{t}\cap V_{t'}| \geq \omega(G)-1$.
  Moreover, as both $V_{t}$ and $V_{t'}$ are maximal (and different),
  we conclude that $|V_{t}|=|V_{t'}|=\omega(G)$ and $|V_{t}\cap V_{t'}| = \omega(G)-1$.
  
  Remember that none of $P$ and $Q$ satisfies condition~$(b)$ of Lemma~\ref{fenceds}.
  So $P$ touches~$V_t$ at least once and~$Q$ touches~$V_{t'}$ at least once.
  First suppose that~$P$ 1-touches~$V_t$ at a vertex~$v$.
  That is, $P$ satisfies condition $(c)$ of Lemma~\ref{fenceds}.
  By (\ref{PQVtVt'}),
  $$\emptyset \ \neq \ P \cap Q \ = \ P \cap V_{t} \cap Q \cap V_{t'} 
              \ = \ \{v\} \cap V_{t'} \cap Q \ = \ \{v\} \cap Q.$$
  So $P \cap Q = \{v\}$.
  That is,~$P$ and~$Q$ only intersect each other at~$v$, 
  which implies that~$v$ divides both longest paths in half.
  Let $P'$ and $P''$ be the two $v$-tails of $P$,
  and let $Q'$ and $Q''$ be the two $v$-tails of $Q$.
  Let $\{u\} = V_{t} \setminus V_{t'}$ and $\{w\} = V_{t'} \setminus V_t$.
  As $P$ 1-touches $V_t$, we may assume without loss of generality that $w \notin P'$.
  Suppose that $Q$ also 1-touches $V_{t'}$.
  Then, we may assume without loss of generality that $u \notin Q'$.
  But then $P' \cdot Q'$ is a longest path that 1-touches $V_t$ at $v$ and crosses~$V_t$.
  As $P$ exists, condition~$(c)$ of Lemma~\ref{fenceds} is not satisfied, a contradiction.
  Now suppose that $Q$ 2-touches $V_{t'}$ at $\{v,x\}$.  Note that $Q_v = Q'$ or $Q_v = Q''$.
  If $u \notin Q_v$ then $P' \cdot Q_v$ is a longest path that 1-touches $V_t$ at $v$ and 
  crosses~$V_t$, again a contradiction.  Hence, $u \in Q_v$. 
  But then $P' \cdot \tilde{Q} \cdot Q_x$ is a longest path that 2-touches $V_{t'}$ and crosses $V_{t'}$.
  As $Q$ exists, condition~$(d)$ of Lemma~\ref{fenceds} is not satisfied, again a contradiction.
  Therefore $P$ touches~$V_t$ at least twice.

  By a similar reasoning, we may conclude that~$Q$ touches $V_{t'}$ at least twice.
  So both~$P$ and~$Q$ must satisfy condition~$(d)$ of Lemma~\ref{fenceds}. 
  Suppose that~$P$ 2-touches~$V_t$ at the ends of edge~$xy$.
  First suppose that $Q$ also 2-touches $V_{t'}$ at the same vertices.
  Then, $|P_x|=|Q_x|$, $|P_y|=|Q_y|$, and $|\tilde{P}|=|\tilde{Q}|$.
  If $u \notin Q_x$ then $P_y \cdot \tilde{P} \cdot Q_x$ is a longest path 
  that 2-touches $V_t$ and crosses $V_t$.
  As $P$ exists, condition $(d)$ of Lemma \ref{fenceds} is not satisfied, a contradiction.
  Hence, $u \in Q_x$ and 
  $u \notin \tilde{Q}$. Then $P_x \cdot \tilde{Q} \cdot P_y$ is a longest path 
  that 2-touches $V_t$ and crosses $V_t$, again a contradiction.
  Hence, we may assume that $Q$ 2-touches $V_{t'}$ at the ends of an edge $yz$ with $z \neq x$. 
  Then $P_x \cdot xz \cdot \tilde{Q} \cdot P_y$
  and $Q_z \cdot zx \cdot \tilde{P} \cdot Q_y$ are paths,
  yielding the final contradiction. 
\end{proof}

The previous theorem implies the following results.

\begin{corolario}\label{corol-chordal1}
  If~$G$ is a tree or a 2-tree, then~$\lpt(G)=1.$
\end{corolario}

\begin{corolario}\label{corol-chordal2}
  If~$G$ is a 3-tree or a connected chordal planar graph, then~$\lpt(G)\leq 2$.
\end{corolario}


\section{Bipartite permutation graphs}\label{section:bpg}

Let $\mathcal{L}_1$ and
$\mathcal{L}_2$ be two parallel lines in the plane.
Consider two sets $X=\{x_1,x_2,\ldots,x_n\}$ and $Y=\{y_1,y_2,\ldots,y_m\}$
of segments that joins a point in $\mathcal{L}_1$ with a point in
$\mathcal{L}_2$, such that the extremes of 
every two elements in $X \cup Y$ are pairwise disjoint.
Moreover, every two elements in~$X$ do not intersect each other and
every two elements in~$Y$ do not intersect each other.

Let~$\sigma$ be the function that maps the extreme in~$\mathcal{L}_1$ of a segment
to the other extreme.
That is, if~$r_i$ is the extreme in~$\mathcal{L}_1$ of a segment
in $X$, then the other extreme is~$\sigma(r_i)$;
and if $s_i$ is the extreme in~$\mathcal{L}_1$ of a segment
in~$Y$, then the other extreme is $\sigma(s_i)$.
Consider an associated bipartite graph~$G=(X, Y, E)$
where~$xy \in E$ if and only if the segments~$x$ and~$y$ intersect each other.
We call the tuple ($\mathcal{L}_1, \mathcal{L}_2, X \cup Y, \sigma)$
a \emph{line representation} of~$G$ and a
graph is called a  
\emph{bipartite permutation graph}
if it has a line representation. (See Figure \ref{fig:bpg-first-example}.)

\begin{figure}[ht]
\centering

\subfigure[ ]{
\begin{tikzpicture}[thick,
  every node/.style={draw,circle,scale=0.9},
  psnode/.style={fill=mygreen},
  qsnode/.style={fill=myblue},
  style1/.style={ellipse,draw,inner sep=-1pt,text width=1.5cm},
  style2/.style={ellipse,draw,inner sep=0pt,text width=4cm,text heigth=16cm,ellipse ratio=2},
]

\begin{scope}[start chain=going right,node distance=7mm, xshift=3cm, yshift=2.0cm]
\node[psnode,on chain] (x1) [label=above: $x_1$] {};
\node[psnode,on chain] (x2) [label=above: $x_2$] {};
\node[psnode,on chain] (x3) [label=above: $x_3$] {};
\end{scope}

\begin{scope}[xshift=3cm, yshift=0.5cm,start chain=going right,node distance=7mm]
\node[qsnode,on chain] (y1) [label=below: $y_1$] {};
\node[qsnode,on chain] (y2) [label=below: $y_2$] {};
\node[qsnode,on chain] (y3) [label=below: $y_3$] {};
\node[qsnode,on chain] (y4) [label=below: $y_4$] {};
\end{scope}

\draw (x1) -- (y1);
\draw (x1) -- (y2);
\draw (x1) -- (y3);
\draw (x2) -- (y1);
\draw (x2) -- (y2);
\draw (x2) -- (y3);
\draw (x2) -- (y4);
\draw (x3) -- (y3);
\draw (x3) -- (y4);
\end{tikzpicture}
}\hspace{15mm}
\subfigure[ ]{
\begin{tikzpicture}[dot/.style={circle,inner sep=1pt,fill,name=#1},
  extended line/.style={shorten >=-#1,shorten <=-#1},
  extended line/.default=1cm]

\node [dot=r1, label=$r_1$] at (1,3) {};
\node [dot=r2, label=$r_2$] at (2,3) {};
\node [dot=s1, label=$s_1$] at (3,3) {};
\node [dot=s2, label=$s_2$] at (4,3) {};
\node [dot=r3, label=$r_3$] at (5,3) {};
\node [dot=s3, label=$s_3$] at (6,3) {};
\node [dot=s4, label=$s_4$] at (7,3) {};

\node [dot=ys1, label=below:$\sigma(s_1)$] at (1,1) {};
\node [dot=ys2, label=below:$\sigma(s_2)$] at (2,1) {};
\node [dot=ys3, label=below:$\sigma(s_3)$] at (3,1) {};
\node [dot=yr1, label=below:$\sigma(r_1)$] at (4,1) {};
\node [dot=ys4, label=below:$\sigma(s_4)$] at (5,1) {};
\node [dot=yr2, label=below:$\sigma(r_2)$] at (6,1) {};
\node [dot=yr3, label=below:$\sigma(r_3)$] at (7,1) {};

\draw [extended line=0.5cm] (r1) -- (s4);
\draw [extended line=0.5cm] (ys1) -- +($(s4)-(r1)$);

\draw [color=mygreen] (r1) -- (yr1);
\draw [color=mygreen] (r2) -- (yr2);
\draw [color=mygreen] (r3) -- (yr3);

\draw [color=myblue]  (s1) -- (ys1);
\draw [color=myblue] (s2) -- (ys2);
\draw [color=myblue] (s3) -- (ys3);
\draw [color=myblue]  (s4) -- (ys4);

\end{tikzpicture}
}

\caption{(a) A bipartite permutation graph. (b) Its corresponding line representation.}
\label{fig:bpg-first-example}
\end{figure}
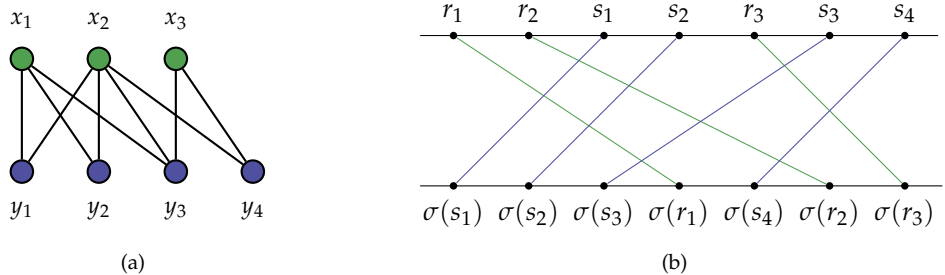

In what follows, we assume that $G=(X,Y,E)$ is a connected bipartite permutation graph, 
with a line representation $(\mathcal{L}_1, \mathcal{L}_2, X \cup Y, \sigma)$, where 
$X=\{x_1,x_2,\ldots,x_n\}$ and $Y=\{y_1,y_2,\ldots,y_m\}$.  We also assume that $r_i$ is 
the extreme of $x_i$ in $\mathcal{L}_1$ and that $s_i$ is the extreme of $y_i$ in $\mathcal{L}_1$.
Moreover, we consider that the sets $X$ and $Y$ are ordered by its extremes in $\mathcal{L}_1$.
That is, $r_i<r_j$ if and only if~$i<j$, for every $x_i$ and $x_j$ in $X$;
and $s_i<s_j$ if and only if $i<j$, for every $y_i$ and $y_j$ in $Y$.
For two elements $x_i$ and $x_j$ in $X$ with $i<j$, we also say that $x_i<x_j$,
and we do the same for $Y$.
Next we show some basic properties of bipartite permutation graphs.
  
\begin{prop}\label{sx1<sx2}
  $\sigma(r_1) < \sigma(r_2) < \cdots < \sigma(r_n)$ and 
  $\sigma(s_1) < \sigma(s_2) < \cdots < \sigma(s_m)$.
\end{prop}  
\begin{proof}
  Suppose by contradiction that there exist $i$ and $j$ 
  such that $r_i < r_j$ and $\sigma(r_i)>\sigma(r_j)$.
  Hence~$x_i$ and~$x_j$ would be adjacent, a contradiction
  to the fact that $G$ is bipartite. A similar proof 
  applies for $Y$ instead of~$X$.
\end{proof}

The next proposition says that the neighborhood of a vertex is 
either completely to its left or completely to its right.

\begin{prop}\label{sameside}
  If $x_i, x_j \in N(y_k)$ for some $i,j,k$, then $(r_i-s_k)(r_j-s_k) > 0$ .
  In a similar way, if~$y_i, y_j \in N(x_k)$ for some $i,j,k$, 
  then $(s_i-r_k)(s_j-r_k) > 0$.  
\end{prop}  
\begin{proof}
  Suppose by contradiction that there exist $i,j,k$ 
  such that $r_i < s_k < r_j$ and $r_i, r_j \in N(s_k)$.
  Then $\sigma(r_j) < \sigma(s_k) < \sigma(r_i)$,
  a contradiction to Proposition \ref{sx1<sx2}. 
  A similar proof applies when considering $Y$ instead of $X$.
\end{proof}

\begin{prop}\label{middle}
  If~$x_i$ is adjacent to~$y_{j_1}$ and~$y_{j_2}$, 
  with~$j_1 \leq j_2$, then~$x_i$ is adjacent to 
  every~$y_j$ with~$j_1 \leq j \leq j_2$.
  If~$x_{i_1}$ and $x_{i_2}$, with $i_1 \leq i_2$,
  are adjacent to~$y_j$, then~$y_j$ is 
  adjacent to every~$x_i$ with~$i_1 \leq i \leq i_2$.
\end{prop}  
\begin{proof}
  The two statements are symmetric, so we only analyze the first one.
  The case in which ${j \in \{j_1, j_2\}}$ is clear, so we may assume that~$j_1 < j < j_2$.
  By Proposition~\ref{sameside}, either ${x_i < y_{j_1} < y_j < y_{j_2}}$ or~$y_{j_1} < y_j < y_{j_2} < x_i$.
  Consider the first case.	
  As~$x_iy_{j_2} \in E$, we have that $\sigma(s_{j_2})< \sigma(r_i)$. 
  Also, $\sigma(s_j)< \sigma(s_{j_2})$ by Proposition \ref{sx1<sx2}.
  So~$\sigma(s_j) < \sigma (r_i)$, implying that~$x_i$ and~$y_j$ are adjacent.
  We can apply a similar argument to deduce that~$y_j$ 
  and~$x_i$ are also adjacent in the second case.
\end{proof}

The following two properties are very important,
as they will be used repeatedly throughout the next proofs.

\begin{prop}\label{cruzados}
  If $x_{i_1}y_{j_2}$, $x_{i_2}y_{j_1} \in E$, with $i_1\leq i_2$ and $j_1 \leq j_2$,
  then $x_{i_1}y_{j_1}$, $x_{i_2}y_{j_2} \in E$. In other words,
  $\{x_{i_1},x_{i_2},y_{j_1},y_{j_2}\}$ induces a complete bipartite graph.
\end{prop}
\begin{proof}
  The case in which $j_1=j_2$ or $i_1=i_2$ is clear, so let us assume 
  that~$i_1 < i_2$ and $j_1 < j_2$.  First suppose that $r_{i_1} < s_{j_2}$.
  Then $\sigma(s_{j_2}) < \sigma(r_{i_1})$ because $x_{i_1}y_{j_2} \in E$.
  By Proposition~\ref{sx1<sx2}, we have that 
  $\sigma(r_{i_1}) < \sigma(r_{i_2})$ and~$\sigma(s_{j_1}) < \sigma(s_{j_2})$.
  So 
  \begin{equation} \label{sigma-s1s2r1r2}
    \sigma(s_{j_1}) < \sigma(s_{j_2}) < \sigma(r_{i_1}) < \sigma(r_{i_2}).
  \end{equation}
  Hence, as $x_{i_2}y_{j_1} \in E$, we have that $r_{i_2} < s_{j_1}$, and
  \begin{equation} \label{r1r2s1s2}
    r_{i_1} < r_{i_2} < s_{j_1} < s_{j_2}.
  \end{equation}
  By \eqref{sigma-s1s2r1r2} and \eqref{r1r2s1s2}, we derive that
  $x_{i_1}y_{j_1}$, $x_{i_2}y_{j_2} \in E$. (See Figure~\ref{fig:cruzados}(a).)

  Now suppose that $r_{i_1} > s_{j_2}$. Then
  \begin{equation} \label{s1s2r1r2}
    s_{j_1} < s_{j_2} < r_{i_1} < r_{i_2}.
  \end{equation}
  As $x_{i_2}y_{j_1} \in E$, we have that $\sigma(r_{i_2})<\sigma(s_{j_1})$.
  Using Proposition~\ref{sx1<sx2}, we deduce that
  \begin{equation} \label{sigma-r1r2s1s2}
    \sigma(r_{i_1}) < \sigma(r_{i_2}) < \sigma(s_{j_1}) < \sigma(s_{j_2}).
  \end{equation}
  By \eqref{s1s2r1r2} and \eqref{sigma-r1r2s1s2}, we derive that
  $x_{i_1}y_{j_1}$, $x_{i_2}y_{j_2} \in E$.  (See Figure~\ref{fig:cruzados}(b).)
\end{proof}

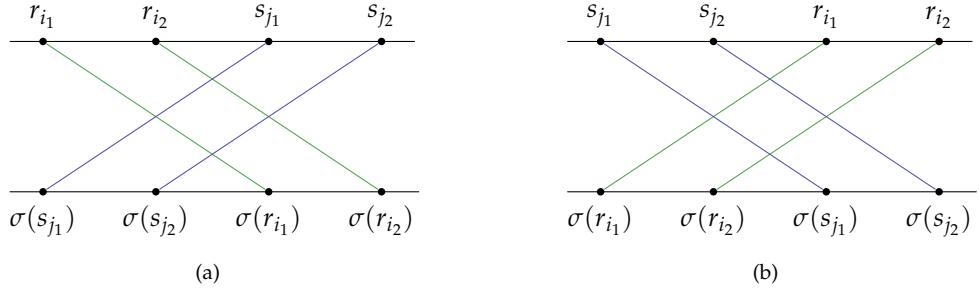
\begin{figure}[ht]
\centering

\subfigure[ ]{

\begin{tikzpicture}[dot/.style={circle,inner sep=1pt,fill,name=#1},
  extended line/.style={shorten >=-#1,shorten <=-#1},
  extended line/.default=1cm]

\node [dot=r1, label=$r_{i_1}$] at (1,3) {};
\node [dot=r2, label=$r_{i_2}$] at (2.5,3) {};
\node [dot=s1, label=$s_{j_1}$] at (4,3) {};
\node [dot=s2, label=$s_{j_2}$] at (5.5,3) {};

\node [dot=ys1, label=below:$\sigma(s_{j_1})$] at (1,1) {};
\node [dot=ys2, label=below:$\sigma(s_{j_2})$] at (2.5,1) {};
\node [dot=yr1, label=below:$\sigma(r_{i_1})$] at (4,1) {};
\node [dot=yr2, label=below:$\sigma(r_{i_2})$] at (5.5,1) {};

\draw [extended line=0.5cm] (r1) -- (s2);
\draw [extended line=0.5cm] (ys1) -- +($(s2)-(r1)$);

\draw [color=mygreen] (r1) -- (yr1);
\draw [color=mygreen] (r2) -- (yr2);

\draw [color=myblue]  (s1) -- (ys1);
\draw [color=myblue] (s2) -- (ys2);
\end{tikzpicture}
}\hspace{15mm}
\subfigure[ ]{
\begin{tikzpicture}[dot/.style={circle,inner sep=1pt,fill,name=#1},
  extended line/.style={shorten >=-#1,shorten <=-#1},
  extended line/.default=1cm]

\node [dot=s1, label=$s_{j_1}$] at (1,3) {};
\node [dot=s2, label=$s_{j_2}$] at (2.5,3) {};
\node [dot=r1, label=$r_{i_1}$] at (4,3) {};
\node [dot=r2, label=$r_{i_2}$] at (5.5,3) {};

\node [dot=yr1, label=below:$\sigma(r_{i_1})$] at (1,1) {};
\node [dot=yr2, label=below:$\sigma(r_{i_2})$] at (2.5,1) {};
\node [dot=ys1, label=below:$\sigma(s_{j_1})$] at (4,1) {};
\node [dot=ys2, label=below:$\sigma(s_{j_2})$] at (5.5,1) {};

\draw [extended line=0.5cm] (s1) -- (r2);
\draw [extended line=0.5cm] (yr1) -- +($(r2)-(s1)$);

\draw [color=mygreen] (r1) -- (yr1);
\draw [color=mygreen] (r2) -- (yr2);

\draw [color=myblue]  (s1) -- (ys1);
\draw [color=myblue] (s2) -- (ys2);
\end{tikzpicture}
}
\caption{The two cases in the proof of Proposition \ref{cruzados}.}
\label{fig:cruzados}
\end{figure}

\begin{prop}\label{cruzados-contraidos}
  Let $x_{i_1} \leq x_{i_2} \leq x_{i_3} \leq x_{i_4}$ be vertices in $X$
  and $y_{j_1} \leq y_{j_2} \leq y_{j_3} \leq y_{j_4}$ be vertices in $Y$.
  If~$x_{i_1}y_{j_4}$ and $x_{i_4}y_{j_1}$ are in $E$, then 
  $x_{i_2}y_{j_3}$, $y_{j_2}x_{i_3}$, $x_{i_2}y_{j_2}$, and $x_{i_3}y_{j_3}$ are in $E$.
   In other words, $\{x_{i_2}, x_{i_3}, y_{j_2}, y_{j_3}\}$ induces a complete bipartite graph.
\end{prop}  
\begin{proof}
  By Proposition~\ref{cruzados}, we have that $x_{i_1}y_{j_1}$ and $x_{i_4}y_{j_4}$ are in $E$.
  By applying Proposition~\ref{middle}, once for $x_{i_1}$ and once for $x_{i_4}$, 
  we deduce that $x_{i_1}y_{j_2}$, $x_{i_1}y_{j_3}$, $x_{i_4}y_{j_2}$, $x_{i_4}y_{j_3}$ are in~$E$, 
  and now applying it for $y_{j_2}$ and for $y_{j_3}$, we deduce that 
  $x_{i_2}y_{j_2}, x_{i_3}y_{j_2}, x_{i_2}y_{j_3}$ and $x_{i_3}y_{j_3}$ are in $E$.
\end{proof}

Until now, we have used the line representation of $G$ to prove some properties.
From now on, we will not need this line representation anymore. 
That is, we only need to 
concentrate in the graph $G$, viewed as a bipartite graph that has the previous 
properties.

We are interested in how do longest paths behave in a bipartite permutation graph.
We begin by showing that every longest path can be converted into another longest 
path with the same set of vertices that is ordered in some way. As we only care about
vertex intersection of longest paths, we will be only interested in such ordered paths.
To be more precise, if $P=a_1b_1a_2b_2\cdots a_kb_k$ is a path in $G$, 
we say that $P$ is \definition{ordered} if $a_1 < a_2 < \cdots < a_k$ and 
$b_1 < b_2 < \cdots < b_k$.  A similar definition applies when $P$ has even length.


Let $P$ be a path in $G$ with $P \cap X =\{a_1,a_2,\ldots, a_{|P \cap X|}\}$ and
$P \cap Y =\{b_1,b_2,\ldots, b_{|P \cap Y|}\}$, so that $a_1 < a_2 < \cdots < a_{|P \cap X|}$
and $b_1 < b_2 < \cdots < b_{|P \cap Y|}$.
For every $i \in \{1, \ldots, |P \cap X|\}$, 
let ${X_i=\{a_1,a_2,\ldots,a_i\}}$ and $\bar{X}_i= P\cap  (X \setminus X_i)$. 
For every $i \in \{1, \ldots, |P \cap Y|\}$,
let ${Y_{i}=\{b_1,b_2,\ldots,b_{i}\}}$ and $\bar{Y}_{i}= P\cap  (Y \setminus Y_i)$.
We denote by~$d_P(X_i)$ the sum $\sum_{v \in X_i}d_P(v)$
and by~$d_P(Y_j)$ the sum $\sum_{w \in Y_j}d_P(w)$.

\begin{prop} \label{connectedness}
  Let $i,j$ be such that $1\leq i \leq |P \cap X|$, $1\leq j \leq |P \cap Y|$, 
  and either $i < |P \cap X|$ or~$j < |P \cap Y|$.
  Then, there exists either an edge in $P$ from $X_i$ to $\bar{Y}_j$,
  or an edge in $P$ from $Y_j$ to $\bar{X}_i$.
\end{prop}
\begin{proof}
  Without loss of generality, assume that $i < |P \cap X|$. 
  Let $v \in P \cap (X \setminus X_i)$. 
  Suppose by contradiction that no such edge exists.
  Then, there is no path, in the subgraph of $G$ induced by~$E(P)$,
  between $v$ and a vertex in $X_i \cup Y_j$, a contradiction.
\end{proof}

\begin{prop}\label{d_P(X_i)-geq-d_P(Y_j)-or-d_P(Y_j)-geq-d_P(X_i)}
  Let $i,j$ be such that $1\leq i \leq |P \cap X|$,
  $1\leq j \leq |P \cap Y|$, and either
  $i < |P \cap X|$ or $j < |P \cap Y|$.
  If $d_P(X_i) \geq d_P(Y_j)$, then
  there exists an edge from $X_i$ to $\bar{Y}_j$.
  If $d_P(Y_j) \geq d_P(X_i)$, then
  there exists an edge from $Y_j$ to $\bar{X}_i$.
\end{prop}
\begin{proof}
  We will prove only the first affirmation, as the proof for the second one
  is analogous.
  Suppose by contradiction that $d_P(X_i) \geq d_P(Y_j)$ and
  there exists no edge from $X_i$ to $\bar{Y}_j$.
  By Proposition \ref{connectedness}, there exists at least one edge from 
  $Y_j$ to $\bar{X}_i$, so
  \begin{eqnarray*}
    d_P(Y_j) &=& \{wv \in E(P): w \in Y_j, v \in X_i\} + \{wv \in E(P): w \in Y_j, v \in \bar{X}_i\} \\
             &=& \{vw \in E(P): v \in X_i, w \in V(P)\} + \{wv \in E(P): w \in Y_j, v \in \bar{X}_i\} \\
             &=& d_P(X_i) + \{wv \in E(P): w \in Y_j, v \in \bar{X}_i\}  \\
             &>& d_P(X_i),
  \end{eqnarray*}
  a contradiction.
\end{proof}

\begin{lema}\label{ordenado}
  For every path $P$ in a bipartite permutation graph,
  there exists an ordered path with the same vertex set as $P$.
\end{lema}
\begin{proof}
  Suppose that $P \cap X =\{a_1,a_2,\ldots, a_{|P \cap X|}\}$,
  $P \cap Y =\{b_1,b_2,\ldots, b_{|P \cap Y|}\}$, 
  $a_1 < \cdots < a_{|P \cap X|}$ and $b_1 < \cdots < b_{|P \cap Y|}$.
  Without loss of generality, we may assume that 
  $|P \cap X| \geq |P \cap Y|$ and that, if $|P \cap X| = |P \cap Y|$,
  then $i^*\leq j^*$, where $a_{i^*}$ is the extreme of $P$ in $X$ and
  $b_{j^*}$ is the extreme of~$P$ in~$Y$. 
  (If $i^*> j^*$, then a similar proof applies by interchanging $X$ and $Y$.)
  Let $k=|P \cap Y|$. We will show that
  \begin{equation} \label{X_i->Y_i-1}
    \mbox{for every } i \in \{1,\ldots, k\}, \mbox{there exists an edge with one end in }X_i \mbox{ and the other in } \bar{Y}_{i-1},
  \end{equation}
  and
  \begin{equation}\label{Y_i->X_i}
    \begin{aligned}
      & \mbox{for every }i \in \{1,\ldots, k-1+|P \cap X|-|P \cap Y|\}, \\
      & \mbox{there exists an edge with one end in }Y_i \mbox{ and the other in }\bar{X}_{i}.
    \end{aligned}
  \end{equation}

  Proof of (\ref{X_i->Y_i-1}): 
  Observe that $d_P(u) = 1$ for at most two vertices~$u$ in~$X_i$ (the extremes of~$P$). 
  Therefore, $d_P(X_i) \geq 2|X_i|-2 = 2|Y_{i-1}|$.
  As $d_P(w) \leq 2$ for every $w \in Y_{i-1}$, we have that $d_P(Y_{i-1}) \leq 2|Y_{i-1}|$. 
  Hence, $d_P(X_i) \geq d_P(Y_{i-1})$ and, as $i-1<k$, 
  $(\ref{X_i->Y_i-1})$ is valid by Proposition \ref{d_P(X_i)-geq-d_P(Y_j)-or-d_P(Y_j)-geq-d_P(X_i)}.

  Proof of (\ref{Y_i->X_i}): 
  First suppose that $|P \cap X|=k$ ($=|P \cap Y|$).
  As $i^*\leq j^*$, we have that~${d_P(Y_i)\geq d_P(X_i)}$.
  Indeed, if $a_{i^*} \in X_i$ then $d_P(X_i)=2|X_i|-2$ and $d_P(Y_i) \geq 2|Y_i|-2$,
  and if $a_{i^*} \notin X_i$ then, as $i^*\leq j^*$, $y_{j^*} \notin Y_i$, so
  $d_P(X_i)=2|X_i|$ and $d_P(Y_i)=2|Y_i|$.
  Thus, as $i<k$, $(\ref{Y_i->X_i})$ is valid by Proposition \ref{d_P(X_i)-geq-d_P(Y_j)-or-d_P(Y_j)-geq-d_P(X_i)}.
  Now suppose that $|P \cap X|=k+1$.
  Then, $d_P(w) = 2$ for every $w \in Y_i$.
  Therefore, $d_P(Y_i) = 2|Y_i| = 2|X_{i}|$.
  As $d_P(v) \leq 2$ for every $v \in X_{i}$, we have that $d_P(X_i) \leq 2|X_i|$. 
  Hence, $d_P(Y_i) \geq d_P(X_i)$ and, as $i<k+1$, 
  $(\ref{Y_i->X_i})$ is valid by Proposition \ref{d_P(X_i)-geq-d_P(Y_j)-or-d_P(Y_j)-geq-d_P(X_i)}.

  Let $i \in \{1,\ldots, k-1\}$.
  By (\ref{X_i->Y_i-1}), there exists a vertex
  $a_{q}$ in~$X_i$ with a neighbor~$b_{r'}$ in~$\bar{Y}_{i-1}$.
  By (\ref{Y_i->X_i}), there exists a vertex $b_{r}$ in~$Y_{i}$ 
  with a neighbor~$a_{q'}$ in~$\bar{X}_i$.
  As $a_q\leq a_{i} \leq  a_{i+1} \leq a_{q'}$ and 
  $b_r \leq b_i \leq b_i \leq b_{r'}$, by Proposition~\ref{cruzados-contraidos}, both $a_ib_i$ 
  and~$b_ia_{i+1}$ are edges. (See Figure~\ref{fig:ordenado-previo}.)
  By (\ref{X_i->Y_i-1}), $a_kb_k$ is an edge, hence $a_1b_1\cdots a_kb_k$ is a path.
  This implies that if $|X\cap P| = k$, we are done.
  Also, if~$|X\cap P| = k+1$, then $b_ka_{k+1}$ is an edge, by (\ref{Y_i->X_i}), 
  so $a_1b_1\cdots a_kb_ka_{k+1}$ is a path.
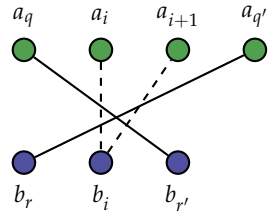
\begin{figure}[ht]
\centering
\begin{tikzpicture}[thick,
  every node/.style={draw,circle,scale=0.9},
  psnode/.style={fill=mygreen},
  qsnode/.style={fill=myblue},
  style1/.style={ellipse,draw,inner sep=-1pt,text width=1.5cm},
  style2/.style={ellipse,draw,inner sep=0pt,text width=4cm,text heigth=16cm,ellipse ratio=2},
]

\begin{scope}[start chain=going right,node distance=7mm, xshift=3cm, yshift=2.0cm]
\node[psnode,on chain] (aq) [label=above: {\makebox[0pt][c]{$a_q$}}] {};
\node[psnode,on chain] (ai) [label=above: {\makebox[0pt][c]{$a_{i_{}}$}}] {};
\node[psnode,on chain] (ai1) [label=above: {\makebox[0pt][c]{$a^{}_{i+1}$}}] {};
\node[psnode,on chain] (aqq) [label=above: {\makebox[0pt][c]{$a_{q'}$}}] {};
\end{scope}

\begin{scope}[xshift=3cm, yshift=0.5cm,start chain=going right,node distance=7mm]
\node[qsnode,on chain] (br) [label=below: {\makebox[0pt][c]{$b^{}_r$}}] {};
\node[qsnode,on chain] (bi) [label=below: {\makebox[0pt][c]{$b^{}_i$}}] {};
\node[qsnode,on chain] (brr) [label=below: {\makebox[0pt][c]{$b^{}_{r'}$}}] {};
\end{scope}

\draw (aq) -- (brr);
\draw (br) -- (aqq);
\draw [style=dashed] (ai) -- (bi);
\draw [style=dashed] (bi) -- (ai1);
\end{tikzpicture}
\caption{The proof of Lemma \ref{ordenado}.}
\label{fig:ordenado-previo}
\end{figure}
\end{proof}

  

As observed before, Lemma \ref{ordenado} says that
we can restrict attention to ordered longest paths from now on.
Remember that we want to prove that~$\lpt(G)=1$. 
We proceed in two steps. First, we will prove that $\lpt(G)\leq 2$.  
In fact, we prove that the set of ends of every edge is a longest path transversal. 
Finally, we will prove that one element in $\{x_1,y_1\}$ is 
also a longest path transversal, which implies that $\lpt(G)=1$.

Let $x_{i_1}y_{j_1}$ and $x_{i_2}y_{j_2}$ be two edges in $G$.
We say that $x_{i_1}y_{j_1}$ and $x_{i_2}y_{j_2}$ \emph{cross} 
each other if 
$(i_1-i_2)(j_1-j_2)<0$.
If that is not the case, we say they are \emph{parallel}.
We say that $|i_1-i_2|$ is the \emph{distance in}~$X$
and that $|j_1-j_2|$ is the \emph{distance in}~$Y$ between such edges.
We denote by $dist_X(x_{i_1}y_{j_1},x_{i_2}y_{j_2})$ and 
$dist_Y(x_{i_1}y_{j_1},x_{i_2}y_{j_2})$ these two values
respectively.

\begin{prop}\label{cross-with-vw}
  Let $P$ be a longest path and $x_{i_1}y_{j_1} \in E(P)$.
  Let $vw \in E(G)$.
  If $x_{i_1}y_{j_1}$ crosses $vw$, then $P$ contains at least one of $\{v,w\}$.
\end{prop}
\begin{proof}
  Suppose by contradiction that $\{v,w\} \cap V(P) = \emptyset$.
  Without loss of generality, suppose that $x_{i_1}<v$ and $y_{j_1}>w$. 
  By Proposition~\ref{cruzados}, $x_{i_1}w$ and $vy_{j_1}$ are edges.
  Therefore, ${P-x_{i_1}y_{j_1}+x_{i_1}w+wv+vy_{j_1}}$ is a path longer than $P$,
  a contradiction.
\end{proof}

\begin{lema} \label{every-edge-is-transversal}
  Let~$G=(X,Y,E)$ be a connected bipartite permutation graph.
  Let~$vw \in E$, with~$v \in X$ and~$w \in Y$.
  Every ordered longest path contains a vertex of~$\{v, w\}$.
\end{lema}

\begin{proof}
  Suppose by contradiction that there exists an ordered
  longest path~$P$ that does not contain either~$v$ or~$w$.
  Then, by Proposition \ref{cross-with-vw}, all edges of~$P$ are parallel to~$vw$.
  Let~$x_{i_1}y_{j_1}$ be the edge of~$P$ that is
  ``closer'' to~$vw$. That is, 
  $dist_X(x_{i_1}y_{j_1}, vw) = \min\{dist_X(e,vw): e\in E(P)\}$
  and~$dist_Y(x_{i_1}y_{j_1}, vw) = \min\{dist_X(e,vw): e\in E(P)\}$.
  Observe that, as~$P$ is an ordered path, one of~$\{x_{i_1},y_{j_1}\}$ is an extreme of~$P$. 
  Suppose that~$x_{i_1}$ is such an extreme.
  (A similar proof applies when this is not the case.) 
  Without loss of generality, we may assume that $x_{i_1} > v$ and that~$P$ 
  is a path with minimum value of~$x_{i_1}$ among all such paths.
  
  Let~$H$ be the subgraph of~$G$ induced by the vertices
  $\{x_i:i\geq i_1\} \cup \{y_j:j\geq j_1\}$.
  As~$G$ is connected and~$G \neq H$, there exists an edge 
  between~$H$ and~$G-V(H)$.  First suppose that such an edge is 
  between a vertex~$x_l$ in~$H$ and a vertex~$y_r$ in $G-V(H)$.
  Then, by Proposition~\ref{cruzados}, $x_{i_1}$ is adjacent to~$y_r$.
  Hence, $y_rx_{i_1} \cdot P$ is also a path, a contradiction.
  (See Figure~\ref{fig:every-edge-is-transversal}(a).)
  Now suppose that there is an edge between a vertex $x_l$ in $G-V(H)$ 
  and a vertex $y_r$ in $H$. Then, by 
  Proposition~\ref{cruzados}, $x_{l}$ is adjacent to~$y_{j_1}$.  
  So $Q=P-x_{i_1}y_{j_1}+x_ly_{j_1}$ is also a longest path.
  As $V(Q) \setminus V(P) = \{x_l\}$, we have that $w \notin Q$.
  Observe also that $v \notin Q$. Indeed, otherwise 
  $Q \cdot vw$ is a path longer than $P$.
  Hence, by Proposition~\ref{cross-with-vw}, all edges of $Q$ are parallel to $vw$,
  which implies that $v<x_l<x_{i_1}$,
  a contradiction to the way~$P$ was chosen.
  (See Figure~\ref{fig:every-edge-is-transversal}(b).)
\end{proof}

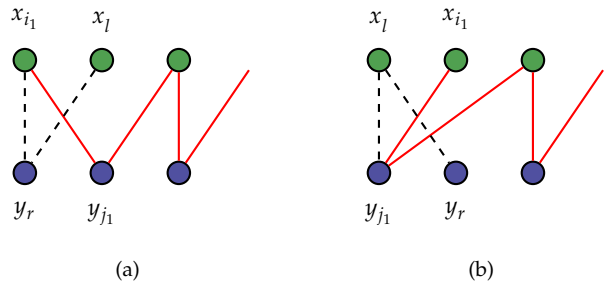
\begin{figure}[ht]
\centering
\subfigure[ ]{
\begin{tikzpicture}[thick,
  every node/.style={draw,circle,scale=0.9},
  psnode/.style={fill=mygreen},
  qsnode/.style={fill=myblue},
  style1/.style={ellipse,draw,inner sep=-1pt,text width=1.5cm},
  style2/.style={ellipse,draw,inner sep=0pt,text width=4cm,text heigth=16cm,ellipse ratio=2},
]

\begin{scope}[start chain=going right,node distance=7mm, xshift=3cm, yshift=2.0cm]
\node[psnode,on chain] (xi1) [label=above: $x_{i_1}$] {};
\node[psnode,on chain] (xl) [label=above: $x^{}_l$] {};
\node[psnode,on chain] (x2) [label=above: $ $] {};
\node[on chain, draw=none] (x3) [label=above: $ $] {};
\end{scope}

\begin{scope}[xshift=3cm, yshift=0.5cm,start chain=going right,node distance=7mm]
\node[qsnode,on chain] (yr) [label=below: $y^{}_r$] {};
\node[qsnode,on chain] (yj1) [label=below: $y_{j_1}$] {};
\node[qsnode,on chain] (y2) [label=below: $ $] {};
\end{scope}

\draw [color = myred] (xi1) -- (yj1);
\draw [color = myred] (yj1) -- (x2);
\draw [color = myred] (x2) -- (y2);
\draw [color = myred] (y2) -- (x3);
\draw [style=dashed] (xl) -- (yr);
\draw [style=dashed] (xi1) -- (yr);
\end{tikzpicture}
}\hspace{8mm}
\subfigure[ ]{
\begin{tikzpicture}[thick,
  every node/.style={draw,circle,scale=0.9},
  psnode/.style={fill=mygreen},
  qsnode/.style={fill=myblue},
  style1/.style={ellipse,draw,inner sep=-1pt,text width=1.5cm},
  style2/.style={ellipse,draw,inner sep=0pt,text width=4cm,text heigth=16cm,ellipse ratio=2},
]

\begin{scope}[start chain=going right,node distance=7mm, xshift=3cm, yshift=2.0cm]
\node[psnode,on chain] (xl) [label=above: $x^{}_l$] {};
\node[psnode,on chain] (xi1) [label=above: $x_{i_1}$] {};
\node[psnode,on chain] (x2) [label=above: $ $] {};
\node[draw=none, on chain] (x3) [label=above: $ $] {};
\end{scope}

\begin{scope}[xshift=3cm, yshift=0.5cm,start chain=going right,node distance=7mm]
\node[qsnode,on chain] (yj1) [label=below: $y_{j_1}$] {};
\node[qsnode,on chain] (yr) [label=below: $y^{}_r$] {};
\node[qsnode,on chain] (y2) [label=above: $ $] {};
\node[on chain, draw=none] (y3) [label=above: $ $] {};
\end{scope}

\draw [color = myred] (xi1) -- (yj1);
\draw [color = myred] (yj1) -- (x2);
\draw [color = myred] (x2) -- (y2);
\draw [color = myred] (y2) -- (x3);
\draw [style=dashed] (xl) -- (yr);
\draw [style=dashed] (xl) -- (yj1);
\end{tikzpicture}
}
\caption{The two cases in the proof of Lemma~\ref{every-edge-is-transversal}.
The solid lines represent the edges in $P$.}
\label{fig:every-edge-is-transversal}
\end{figure}

Given a collection~$\calC$ of ordered longest paths, we say that 
$P\in \calC$ is a \emph{left-most} path if, for every other path
$Q\in \calC$ and for every $i$, the $i$-th vertex of~$P$ 
in~$X$ is less than or equal to the $i$-th vertex of~$Q$ in~$X$, and 
the same applies for~$Y$ instead of~$X$.  Such a path exists 
because all paths in~$\calC$ are ordered.
		


\begin{teorema}\label{lpt(bpg)=1}
  For every connected bipartite permutation graph $G=(X,Y,E)$, $\lpt(G)=\nobreak1$.
\end{teorema}
\begin{proof}
  Suppose by contradiction that $\lpt(G)>1$. Then, there exists a longest 
  path~$P$ that does not contain~$y_1$ and a longest path $Q$ that does not 
  contain~$x_1$. As~$G$ is connected, $x_1y_1$ is an edge by Proposition~\ref{cruzados}.
  So, by Lemma \ref{every-edge-is-transversal}, $x_1 \in P$ and $y_1 \in Q$.  
  We may assume that both~$P$ and~$Q$ are left-most paths.
  Suppose without loss of generality that $n\geq m$.  Thus, for all 
  $i \in \{2, \ldots, m \}$, it suffices to prove the following conditions:
  \begin{itemize}
  \item[(a)] $y_i$ is the $(2i-3)$-th vertex of $P$,
             $x_{i-1}$ is the $(2i-2)$-th vertex of $P$,     
             $x_i$ is the $(2i-3)$-th vertex of $Q$, and
             $y_{i-1}$ is the $(2i-2)$-th vertex of $Q$.
  \item[(b)] $x_iy_i$ is an edge.
  \end{itemize}
  Indeed, if that is the case, then we would have a 
  path~$R = (x_1y_1x_2y_2\cdots x_my_m)$ of length~$2m-1$.
  As~$P$ does not contain~$y_1$, we would have 
  $|P \cap Y| = |P \cap \{y_2, y_3,\ldots,y_m\}| = m-1$.
  And, as~$G$ is bipartite,~$|P\cap X|\leq m$. 
  Hence $|P| \leq 2m-2 < |R|$, a contradiction, 
  because~$P$ is a longest path. 

\begin{figure}[ht]
\centering
\begin{tikzpicture}[thick,
  every node/.style={draw,circle,scale=0.9},
  psnode/.style={fill=mygreen},
  qsnode/.style={fill=myblue},
  style1/.style={ellipse,draw,inner sep=-1pt,text width=1.5cm},
  style2/.style={ellipse,draw,inner sep=0pt,text width=4cm,text heigth=16cm,ellipse ratio=2},
]

\begin{scope}[start chain=going right,node distance=7mm, xshift=3cm, yshift=2.0cm]
\node[psnode,on chain] (x1) [label=above: $x_1$] {};
\node[psnode,on chain] (x2) [label=above: $x_2$] {};
\node[psnode,on chain] (x3) [label=above: $x_3$] {};
\node[psnode,on chain] (x4) [label=above: $x_4$] {};
\node[on chain, draw=none] (xn) [label=above: $ $] {};
\end{scope}

\begin{scope}[xshift=3cm, yshift=0.5cm,start chain=going right,node distance=7mm]
\node[qsnode,on chain] (y1) [label=below: $y_1$] {};
\node[qsnode,on chain] (y2) [label=below: $y_2$] {};
\node[qsnode,on chain] (y3) [label=below: $y_3$] {};
\node[qsnode,on chain] (y4) [label=below: $y_4$] {};
\node[on chain, draw=none] (yn) [label=below: $ $] {};
\end{scope}

\draw  [color = mybrown]  (y2) -- (x1);
\draw  [color = mybrown] (x1) -- (y3);
\draw  [color = mybrown] (y3) -- (x2);
\draw  [color = mybrown] (x2) -- (y4);
\draw  [color = mybrown] (y4) -- (x3);
\draw  [color = mybrown] (x3) -- (yn);
\draw [color = myred] (x2) -- (y1);
\draw [color = myred] (y1) -- (x3);
\draw [color = myred] (x3) -- (y2);
\draw [color = myred] (y2) -- (x4);
\draw [color = myred] (x4) -- (y3);
\draw [color = myred] (y3) -- (xn);

\draw [style=dashed] (x1) -- (y1);
\draw [style=dashed] (x2) -- (y2);
\draw [style=dashed] (x3) -- (y3);
\draw [style=dashed] (x4) -- (y4);
\end{tikzpicture}
\label{fig:lpt(bpg)=1}
\end{figure}
	
  We proceed by induction on $i$.  If~$i=2$, we need to prove 
  that $y_2x_1$ and~$x_2y_1$ are the first edges of~$P$ and~$Q$ 
  respectively, and that $x_1y_1, x_2y_2$ are edges.
  Remember that $x_1 \in P$.
  Obviously, $x_1$ is not an extreme of~$P$. So, as~$P$ is 
  an ordered longest path, $x_1$ is the second vertex of~$P$.
  Now we will prove that~$y_2$ is the first vertex of~$P$. 
  If~$P$ starts in $y_j$ with~$j>2$ then, as~$y_jx_1$ 
  and~$x_1y_1$ are edges, $x_1y_2$ is an edge by Proposition~\ref{middle}. 
  Thus $P-x_1y_j+x_1y_2$ is also a longest path, contradicting the choice of~$P$.
  A similar reasoning shows that~$x_2y_1$ is the first edge of~$Q$.
  This implies, by Proposition \ref{cruzados}, that~$x_2y_2$
  is an edge, finishing the base case of the induction.
  
  Now fix an $i>2$ and assume that both~(a) and~(b) are 
  valid for all $j<i$.  Then, by the induction hypothesis, 
  $y_{i-1}x_{i-2}$ is the $(2i-5)$-th edge of~$P$.  
  First, we will prove that~$x_{i-1}$ is the~$(2i-2)$-th vertex of $P$.  
  Indeed, suppose that $x_j$ is the~$(2i-2)$-th vertex of $P$ with $j > i-1$.
  Let~$P = P' \cdot P''$, where $P' \cap P'' = \{x_{i-2}\}$.
  Then $y_1x_1y_2x_2 \cdots y_{i-2}x_{i-2} \cdot P''$ is an 
  ordered longest path that does not contain any vertex of
  $\{x_{i-1}, y_{i-1}\}$, a contradiction to Lemma~\ref{every-edge-is-transversal}.
  So~$x_{i-1}$ is the~$(2i-2)$-th vertex of~$P$. 
  Now, we will prove that $y_i$ is the $(2i-3)$-th vertex of~$P$. 
  Suppose that~$y_j$ is the $(2i-3)$-th vertex of~$P$ with $j>i$.
  Then, by Proposition \ref{cruzados-contraidos}, as, by the induction 
  hypothesis, $x_{i-1}y_{i-1}$ is an edge, $x_{i-2}y_i$ and $x_{i-1}y_i$ are 
  edges.  Now, $P-x_{i-2}y_jx_{i-1}+x_{i-2}y_ix_{i-1}$ is also a longest path,
  contradicting the choice of~$P$.
  A similar argument shows that $x_i$ is the $(2i-3)$-th vertex of~$Q$ 
  and that~$y_{i-1}$ is the $(2i-2)$-th vertex of~$Q$.  This implies, by 
  Proposition~\ref{cruzados}, that~$x_iy_i$ is an edge, finishing the proof.
\end{proof}


\section{Graphs of bounded treewidth and planar graphs}\label{section:tw}

Rautenbach and Sereni~\cite{RautenbachS14} proved that~$\lpt(G) \leq \tw(G) + 1$ for every connected graph $G$.
In this section, we improve their result.

\begin{lema}\label{V_t-transversal}
  Let $G$ be a connected graph.
  Let~$(T, \calV)$ be a 
  tree decomposition of~$G$.
  There exists a node $t \in V(T)$ such that $V_t$ is a longest path transversal.
\end{lema}
\begin{proof}
  Suppose by contradiction that this is not the case.
  Then, for every $t \in V(T)$, there exists a longest path $P$ that does not
  touch $V_t$ (hence $P$ is fenced by $V_t$).
  By Proposition \ref{Branch_t(P)=Branch_t(t')}, there exists a
  neighbor $t'$ of $t$ in $T$ such that $\Branch_t(P) = \Branch_t(t')$.
  We create a directed graph $D$, 
  that admits antiparalell arcs, as follows.  
  The nodes of~$D$ are exactly the nodes of $T$.
  Given a node~$t$ and a neighbor $t'$ of $t$ as before,
  we add~$tt'$ as an arc in~$D$.
  Note that every node of $D$ is the tail of some arc in $D$.
  Let~$tt'$ be the last arc of a maximal directed path in~$D$.
  As~$T$ is a tree,~$t't$ is also an arc in~$D$, which implies 
  that there exist two longest paths~$P$ and~$Q$ in $G$ such that
  $\Branch_t(P)=\Branch_t(t')$ and $\Branch_{t'}(Q)=\Branch_{t'}(t)$,
  where~$P$ is fenced by~$V_t$ and~$Q$ is fenced by~$V_{t'}$,~$P$
  does not touch~$V_t$ and~$Q$ does not touch~$V_{t'}$.
  But then, as the bags containing vertices of~$P$ are only 
  in $\Branch_t(t') \cup \{t\}$, and the bags containing
  vertices of $Q$ are only in $\Branch_{t'}(t) \cup \{t'\}$, 
  the paths~$P$ and~$Q$ do not intersect, a contradiction.
\end{proof}

\begin{teorema}\label{lpt<tw}
  For every connected graph $G$ with treewidth $k$, $\lpt(G) \leq k$.
\end{teorema}

\begin{proof}
  Let~$(T, \calV)$ be a tree decomposition of~$G$ as in Proposition~\ref{tree-dec-bodl}. 
  By Lemma~\ref{V_t-transversal}, there exists a node $t$ in $T$ 
  such that~$V_t$, of size $k+1$, is a longest path transversal.  
  Suppose by contradiction that $\lpt(G) > k$. 
  Then no set of~$k$ vertices in~$V_t$ is a longest path transversal.  
  As every longest path touches~$V_t$ at least once, for every vertex in~$V_t$, 
  there exists a longest path that 1-touches~$V_t$ at that vertex.
  Let $P$ be a longest path that touches $V_t$ at $x$, let $P'$
  and $P''$ be the two $x$-tails of $P$. We will show that
  \begin{equation}\label{Branch_t(P')-neq-Branch_t(P'')}
  \Branch_t(P') \neq \Branch_t(P''). 
  \end{equation}    
  Proof of (\ref{Branch_t(P')-neq-Branch_t(P'')}): Suppose by
  contradiction that $\Branch_t(P')= \Branch_t(P'')=\Branch_t(t')$.
  By Proposition~\ref{tree-dec-bodl}, 
  there exists a vertex~$y$ in~$V_t \setminus V_{t'}$.
  Let~$Q$ be a longest path that 1-touches~$V_t$ at~$y$.
  Let~$Q'$ and $Q''$ be the two $y$-tails of $P$.
  By Proposition \ref{V(P')-cap-V(Q')=emptyset}, both $P \cap Q'$ and $P \cap Q''$
  are empty, a contradiction.
  
  By (\ref{Branch_t(P')-neq-Branch_t(P'')}), there exist two different nodes $t'$ and 
  $t''$ that are adjacent to~$t$ in~$T$ such that~$t'$ is in $\Branch_{t}(P')$ and~$t''$ is in $\Branch_{t}(P'')$.
  By Proposition~\ref{tree-dec-bodl}, there exists a vertex~$a$ 
  in $V_t \setminus V_{t'}$ and a vertex~$b$ in $V_t \setminus V_{t''}$.
  As $t'\neq t''$, we have that $V_{t'}\neq V_{t''}$ and $a \neq b$.
  Let~$Q$ and~$R$ be corresponding longest paths that 1-touch~$V_t$ at~$a$ and~$b$ respectively.  
  By~(\ref{Branch_t(P')-neq-Branch_t(P'')}), both $P$ and $Q$ cross~$V_t$. 
  Observe that $x \in V_{t'} \cap V_{t''}$ by Proposition~\ref{Vt capP'subseteqVt'}, 
  and hence $a \neq x \neq b$.
  
  Let $Q'$ and $Q''$ be the two $a$-tails of $Q$, and let $R'$ and~$R''$ be the two $b$-tails of $R$.  
  By Proposition~\ref{V(P')-cap-V(Q')=emptyset}, paths $P'$ and $Q$ do not intersect.
  So, as~$G$ is connected, $P''$ intersects $Q$.
  Since ${P'' \cap Q \cap V_t = \emptyset}$, we may assume, without loss of generality,
  that~$Q''$ intersects~$P''$, thus~$\Branch_t(Q'')=\Branch_t(P'')=\Branch_t(t'')$.
  Analogously, with a similar analysis with~$R$ instead of~$Q$, we may assume that
  $\Branch_t(R')=\Branch_t(P')=\Branch_t(t')$.
  Applying (\ref{Branch_t(P')-neq-Branch_t(P'')}) with~$Q$ and~$R$ instead of~$P$, 
  one can show that $\Branch_t(Q') \neq \Branch_t(Q'')$ and that ${\Branch_t(R') \neq \Branch_t(R'')}$.
  Thus, $Q'$ is disjoint from~$P$, and~$R''$ is disjoint from~$P$.
  Also, as~$\Branch_t(R')=\Branch_t(t')$,
  by Proposition \ref{V(P')-cap-V(Q')=emptyset}, paths~$Q$ and~$R'$ do not intersect.
  Analogously,~$R$ and~$Q''$ do not intersect.
  
  Let $a' \in P'' \cap Q''$ be such that the subpath of $P$ with extremes $x$ and $a'$ is internally 
  disjoint from $Q''$.
  Let $Q_1$ and $Q_2$ be the two $a'$-tails of $Q$, with $Q_1$ containing $a$.  
  Let $b' \in P' \cap R'$ be such that the subpath of $P$ with extremes $x$ and $b'$ is internally 
  disjoint from $R'$.
  Let $R_1$ and $R_2$ be the two $b'$-tails of $R$, with $R_1$ containing $b$.  
  Let $\tilde{P}$ be the subpath of $P$ that has $a'$ and $b'$ as extremes.
  As $\tilde{P}$ is internally disjoint from both $Q$ and $R$, we have that
  $R_1 \cdot \tilde{P}\cdot Q_2$ and $Q_1\cdot \tilde{P}\cdot R_2$ 
  are paths whose lengths sum more than $2L$, a contradiction.
\end{proof}

The graph of Figure~\ref{waltherCounterexample} has treewidth two. 
Hence, we have the following corollary.

\begin{corolario}\label{series-parallel}
If $G$ is a connected partial 3-tree, then $\lpt(G) \in \{2,3\}$.
\end{corolario}

Planar graphs do not have bounded treewidth.
However, Fomin and Thilikos~\cite{FominT06} showed that a planar graph $G$ 
on $n$ vertices has treewidth at most $3.182 \sqrt{n}$.
More generally,
Alon, Seymour, and Thomas \cite{AlonST90} showed that any
$K_r$-minor free graph on $n$ vertices has treewidth at most $r^{1.5}\sqrt{n}$.
Hence, we have the following corollaries.
The first of them improves the upper bound given by 
Rautenbach and Sereni when the graph is planar.

\begin{corolario}\label{cor1-lpt<tw}
  For every connected planar graph $G$ on $n$ vertices, $\lpt(G) \leq 3.182 \sqrt{n}$.
\end{corolario}

\begin{corolario}\label{cor2-lpt<tw}
  For every connected $K_r$-minor free graph $G$, $\lpt(G) \leq r^{1.5}\sqrt{n}$.
\end{corolario}







\section{Full substar graphs} \label{section:substars}

A \emph{star} is a complete bipartite graph $K_{1,k}$, for some integer $k$.  
If $k \geq 2$, we call the unique vertex of degree $k$ the \emph{center} of the star. 
If $k=1$, we pick an arbitrary vertex to be the center of the star.
Given a tree $T$, a subgraph of $T$ that is a star is called a \emph{substar of T}. 
We say that a star in $T$ with center in $x$ is a \emph{full substar} of $T$ if $k \geq d_T(x)-1$.   
A graph is a \emph{full substar graph} if it is the intersection graph of a set of full substars of a tree. 
In the intersection model, we call $S_x$ the substar of the host tree associated 
with $x \in V(G)$.  We use capital letters to refer to the vertices of the host 
tree of the intersection model and lowercase letters to refer to the vertices of 
the intersection graph.  It can be seen from the definition that every full substar 
graph is also a chordal graph, since chordal graphs are the intersection graphs 
of subtrees of a tree. An example of a full substar graph can be seen in Figure~\ref{fig:hajos}. 

\begin{figure}[htb]
\centering
\scalebox{1}{
%
%
\psscalebox{1.0 1.0} 
{
\begin{pspicture}(0,-1.9492788)(12.919072,1.9492788)
\psline[linecolor=black, linewidth=0.08](1.04,-0.25072113)(1.84,0.9492789)(3.44,-1.4507211)(0.24,-1.4507211)(1.04,-0.25072113)(2.64,-0.25072113)(1.84,-1.4507211)(1.04,-0.25072113)
\pscircle[linecolor=black, linewidth=0.06, fillstyle=solid,fillcolor=red, dimen=outer](0.22,-1.4307212){0.22}
\pscircle[linecolor=black, linewidth=0.06, fillstyle=solid,fillcolor=orange, dimen=outer](1.82,-1.4307212){0.22}
\pscircle[linecolor=black, linewidth=0.06, fillstyle=solid,fillcolor=blue, dimen=outer](2.62,-0.23072113){0.22}
\pscircle[linecolor=black, linewidth=0.06, fillstyle=solid,fillcolor=yellow, dimen=outer](3.42,-1.4307212){0.22}
\pscircle[linecolor=black, linewidth=0.06, fillstyle=solid,fillcolor=violet, dimen=outer](1.82,0.9692789){0.22}
\pscircle[linecolor=black, linewidth=0.06, fillstyle=solid,fillcolor=green, dimen=outer](1.02,-0.23072113){0.22}
\psline[linecolor=black, linewidth=0.08](10.24,-0.65072113)(7.84,-1.8507211)
\psline[linecolor=black, linewidth=0.08](12.64,-1.8507211)(10.24,-0.65072113)(10.24,1.7492789)
\psdots[linecolor=black, dotsize=0.2](10.24,0.54927886)
\psdots[linecolor=black, dotsize=0.2](10.24,-0.65072113)
\psdots[linecolor=black, dotsize=0.2](11.42,-1.2507211)
\psdots[linecolor=black, dotsize=0.2](12.64,-1.8507211)
\psdots[linecolor=black, dotsize=0.2](9.02,-1.2707211)
\psdots[linecolor=black, dotsize=0.2](7.84,-1.8507211)
\psdots[linecolor=black, dotsize=0.2](10.24,1.7492789)
\psline[linecolor=orange, linewidth=0.08](9.04,-1.5107211)(10.24,-0.8707211)(11.38,-1.4507211)
\psline[linecolor=blue, linewidth=0.08](10.44,0.66927886)(10.44,-0.5307211)(11.56,-1.0907211)
\psline[linecolor=green, linewidth=0.08](8.82,-1.1107211)(10.02,-0.5307211)(10.02,0.66927886)
\psline[linecolor=violet, linewidth=0.08](10.62,0.40927887)(10.62,1.9492788)
\psline[linecolor=yellow, linewidth=0.08](11.5,-0.9107211)(12.9,-1.6707212)
\psline[linecolor=red, linewidth=0.08](7.52,-1.6307211)(9.02,-0.8507211)
\end{pspicture}
}}
\caption{A full substar graph.}
\label{fig:hajos}
\end{figure}
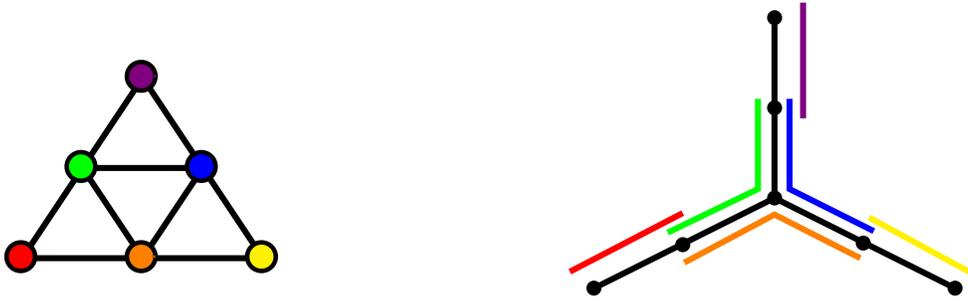 

The definition of branch by Heinz~\cite{Heinz13} applies naturally to an arbitrary tree. 
Here we use it, as well as its variants introduced in Section~\ref{section: definitions}, 
to the tree $T$.
For a vertex $X \in V(T)$, let $\mathcal{C}_X$ be the set of vertices of $G$ 
whose corresponding stars are centered in $X$ and $\mathcal{C}^X_Y$ be the set of 
vertices of~$G$ whose stars are centered in a vertex that belongs to $\Branch_X(Y)$.

In what follows, $G$ is a full substar graph and $T$ is the host tree of an intersection model for~$G$.

\begin{lema} \label{previous-onebranch}
  Let $x \in V(G)$ be such that $x \in \mathcal{C}_X$.
  If~$P$~is a longest path in $G$ such that $x \notin V(P)$,
  then there exists a node $Y \in N_T(X)$   
  such that the following conditions hold:
\begin{enumerate}[(i)]
\item $V(P) \subseteq \mathcal{C}^X_Y \cup \mathcal{C}_X,$
\item $|V(P) \cap \mathcal{C}_X| \leq 1.$
\end{enumerate}
  Moreover, if $|V(P) \cap \mathcal{C}_X| = 1$, then $Y \notin S_x$.
\end{lema}
\begin{proof}
  Let $x$ and $P$ be as stated above.
  First suppose that $V(P)\cap \mathcal{C}_X= \emptyset$.
  Suppose by contradiction that $(i)$ is false.
  Then $P$ has vertices whose substars are centered in two
  different branches of $T$ with respect to $X$. Since $P$ contains no vertex 
  of $\mathcal{C}_X$, then $P$ must contain two consecutive 
  vertices whose stars are centered in the neighborhood of $X$
  in $T$. That is, $P$ contains two consecutive vertices that are
  adjacent to $x$, a contradiction.
  This implies that $(i)$ holds in this case and, since $|V(P)\cap \mathcal{C}_X|=0$,
  $(ii)$ also holds.
  
  Now assume that $V(P)\cap \mathcal{C}_X\neq \emptyset$,
  and let $x'$ be a vertex in $V(P)\cap \mathcal{C}_X$. 
  Suppose by a contradiction that $(i)$ does not hold.
  Then $P$ has vertices whose substars are centered in two
  different branches of $T$ with respect to $X$.
  If $P$ contains two consecutive 
  vertices whose stars are centered in
  different branches, we are in the previous case.
  Then there exists two vertices $u$ and $v$ in $P$ such that
  $ux', vx'\in E(P)$. Since $|S_{x'}|\geq d_T(X)-1$, 
  one of $\{u,v\}$ is  adjacent to $x$, a contradiction. We can conclude that $(i)$ holds in this case. Let $Y$ be as stated in $(i)$.
  
  For showing $(ii)$, suppose that  $|V(P)\cap \mathcal{C}_X|>1$.
  Let $x',x'' \in V(P)\cap \mathcal{C}_X$.
  If $x'x'' \in E(P)$, then we can add $x$ to $P$ 
  between these two vertices.
  In this case, we may assume that 
  $P=P_{x'}\cdot \tilde{P} \cdot P_{x''}$,
  where $V(P_{x'}) \cap V(\tilde{P})= \{x'\}$ and
  $V(P_{x''}) \cap V(\tilde{P}) = \{x''\}$.
  Observe that $x'$ and $x''$ are not extremes of $P$,
  which implies that $P_{x'}$ and $P_{x'}$ are not empty.
  Thus, there exists a vertex $u$ adjacent to $x'$
  in $P_{x'}$ and a vertex $v$ adjacent to $x''$
  in $P_{x''}$.
  Also, there exists a vertex 
  $w'$ adjacent to $x'$
  in $\tilde{P}$
  and a vertex
  $w''$ adjacent to $x''$
  in $\tilde{P}$ (possibly, $w'=w''$).
  Note that $S_{w'} \cap S_{x'} = \{Y\}$
  and $S_{w''} \cap S_{x''} = \{Y\}$,
  implying that $\{x',w',x'',w''\}$ induces a clique.
  By a similar argument,
  $w''$ and $v$ are adjacent.


  In this case, we can find a path longer than $P$ in $G$. Let $x_1$ and $x_k$ be the extremes of $P$ and, given $x,y\in V(P)$, let $P_{xy}$ be the subpath of $P$ that has $x$ and $y$ as its extremes. The path $P_{x_1x'}\cdot x'x\cdot xx''\cdot x''w'\cdot \tilde{P}_{w'w''}\cdot w''v\cdot P_{vx_k}$ is longer than $P$, a contradiction.
  
  To finish the proof,
  suppose by contradiction that
  $|V(P)\cap \mathcal{C}_X|=1$
  and $Y\in S_x$.
  Let $\{x'\} = V(P )\cap \mathcal{C}_X$.
  Since $|P|\geq 1$ and $V(P)\subseteq \mathcal{C}_Y^X\cup \mathcal{C}_X$, there exists an edge $x'v$
  in $P$ such that
  $v\in \mathcal{C}^X_Y$.
  This implies that $x$ is adjacent to both $x'$ and $v$,
  a contradiction.
\end{proof}

\begin{lema} \label{onebranch}
 Let $G$ be a connected full substar graph, $T$ be the host tree of an intersection model for $G$ and let $X$ be any vertex of $T$.
If $\lpt(G)>1$, then there exists a longest path $P$ in $G$ and a node $Y\in N_T(X)$ such that $V(P)\subseteq \mathcal{C}^X_Y$.

\end{lema}
\begin{proof}
We divide the proof in two cases, according to whether there exists a vertex in $G$ such that its corresponding substar
is centered in $X$.

\vspace{0.1in} 
\noindent \emph{Case 1.} $\mathcal{C}_X\neq\emptyset$


 
Let $x\in \mathcal{C}_X$.
Moreover, suppose that $|S_x|$ is maximum
over all such $x$. Since $\lpt(G)>1$, there exists a longest path $P$ in $G$ such that
$x \notin V(P)$.
By Lemma~\ref{previous-onebranch}, there exists a node $Y$,
adjacent to $X$ in $T$
such that
$V(P)\in \mathcal{C}^X_Y \cup \mathcal{C}_X$.
If $V(P) \cap \mathcal{C}_X = \emptyset$, the statement holds.
Otherwise, $P$ has a vertex $x'$ such that $x'\in \mathcal{C}_X$.
Also by Lemma~\ref{previous-onebranch}, $V(P) \cap \mathcal{C}_X=\{x'\}$.
Note that $x'$ is not an extreme of $P$, since $x\notin V(P)$. 
Moreover, if $N_G(x') \subseteq N_G(x)$, then $P$ 
would have to contain $x$. This implies 
that $d_T(X) - 1 \leq |S_{x'}| \leq |S_x| < d_T(X)$ and, as consequence, $|S_x|=|S_{x'}|=d_T(X)-1$.
That is, both $S_x$ and $S_{x'}$ 
miss a node in the neighborhood of $X$.
By Lemma~\ref{previous-onebranch},
$Y \notin S_x$. Since $N_G(x') \not \subseteq N_G(x)$,
we may assume that there exists $Z\in N_T(X)$ such that $Z\neq Y$, $Z\in S_x$ and
$Z \notin S_{x'}$.

Since $\lpt(G)>1$, there exists a longest
path $Q$ in $G$ that does not contain
$x'$.
By Lemma~\ref{previous-onebranch},
$V(Q)\subseteq \mathcal{C}^X_Z \cup \{x''\}$,
for some $x'' \in \mathcal{C}_X$.
However, this implies that $P$ and $Q$ do not intersect each other,
a contradiction with the fact that $G$ is connected.

\vspace{0.1in}
\noindent \emph{Case 2.} $\mathcal{C}_X=\emptyset$

Let $K$ be the clique of $G$ formed by the vertices $x\in V(G)$ such that $X\in S_x$. We will show that if, for every longest path $P$, there is no $Y\in V(T)$ such that $V(P)\subseteq \mathcal{C}^X_Y$, then $\lpt(G)=1$. Suppose that every longest path $P$ of $G$ contains vertices whose substars are centered in two different branches of $T$ with respect to $X$. Since $\mathcal{C}_X=\emptyset$, $P$ must contain two consecutive vertices whose stars are centered in the neighborhood of $X$ in $T$. That is, $P$ has two consecutive vertices that belong to $K$ and therefore $P$ must contain all the vertices of $K$. 
\end{proof}

We are now ready to prove the main result of this section.


\begin{teorema}\label{lpt(alfsbstr)=1}
  If $G$ is a connected full substar graph, then $\lpt(G)=1$.
\end{teorema}

\begin{proof}
Suppose by a contradiction that $\lpt(G)>1$. Let $T$ be the host tree of an intersection model for $G$. We start by creating an auxiliary directed graph $D$ on the same vertex set as $T$ and arc 
set defined in the following way. For every $X\in V(D)$, we have that $XY\in E(D)$ if $Y\in N_T(X)$ and there exists a longest path $P$ such that $V(P)\subseteq \mathcal{C}^X_Y$. By Lemma~\ref{onebranch}, every node in $T$ has outdegree at least one. 

  Let $XY$ be the last arc in a maximal directed path in~$D$.
  Since~$T$ is a tree, $YX$ is also an arc in~$D$. Since $XY\in E(G)$ and $YX\in E(G)$, 
  there exists two longest paths~$P$ and~$Q$ in $G$ such 
  that $V(P)\in \mathcal{C}^X_Y$ and $V(Q)\in \mathcal{C}^Y_X$.
  However, since $\mathcal{C}^Y_X \cap \mathcal{C}^X_Y=\emptyset$, the paths $P$ and $Q$ do not have a vertex in common, a contradiction with 
  the fact that $G$ is connected.

\end{proof}



\section{Conclusion and future work} \label{section:conclusion}

The problem of finding a minimum longest path transversal remains open for several well-studied graph classes. In this work, we proved that connected bipartite permutation graphs admit a transversal of size one. The problem remains open for connected 
biconvex graphs and connected permutation graphs, well-known superclasses of bipartite permutation graphs. Even though our upper bound for $\lpt(G)$, when $G$ is a connected chordal graph, depends on~$\omega(G)$, so far there are no examples of connected chordal graphs that require a transversal of size greater than one. 
In this direction, one open problem is to look for such an example, if it exists, or to look for better bounds for $\lpt(G)$ when $G$ belongs to this graph class. 
Finally, it would be interesting to generalize Theorem~\ref{lpt(alfsbstr)=1} for the class of substar graphs, that is, intersection graphs of substars of a tree.

\bibliographystyle{plain}

  

\bibliographystyle{plain}




\end{document}